\documentclass{article}

\usepackage{./sty/leosstylefile}
\usepackage{./sty/tikzit}
\usetikzlibrary{shapes}
\usetikzlibrary{arrows}

\tikzstyle{node}=[fill=black, draw=black, shape=circle, scale=0.5]
\tikzstyle{wnode}=[fill=white, draw=black, shape=circle, scale=0.5]
\tikzstyle{textbox}=[inner sep=2pt, shape=rectangle, fill=none]
\tikzstyle{textnode}=[inner sep=0mm, shape=circle, fill=white]
\tikzstyle{gnode}=[inner sep=0mm, minimum size=1mm, fill={rgb,255: red,221; green,221; blue,221}, draw={rgb,255: red,221; green,221; blue,221}, shape=circle]
\tikzstyle{refine}=[fill=black, draw=black, shape=regular polygon, regular polygon sides=3, rotate=180, scale=0.5]
\tikzstyle{coarsen}=[fill=white, draw=black, shape=regular polygon, regular polygon sides=3, scale=0.5]
\tikzstyle{bdytextbox}=[fill=white, draw=black, shape=rectangle]
\tikzstyle{redbox}=[fill=white, draw=red, shape=rectangle, text=red]
\tikzstyle{bluecirc}=[inner sep=1mm, fill=white, draw={rgb,255: red,4; green,51; blue,255}, shape=circle, text={rgb,255: red,4; green,51; blue,255}]
\tikzstyle{rednode}=[fill={rgb,255: red,255; green,128; blue,128}, draw={rgb,255: red,255; green,128; blue,128}, shape=circle]
\tikzstyle{new style 0}=[fill=white, draw=red, shape=circle]
\tikzstyle{bluenode}=[fill={rgb,255: red,125; green,221; blue,255}, draw={rgb,255: red,123; green,202; blue,255}, shape=circle]
\tikzstyle{yellownode}=[fill={rgb,255: red,255; green,210; blue,75}, draw={rgb,255: red,255; green,210; blue,75}, shape=circle]
\tikzstyle{blacksq}=[fill=black, draw=black, shape=rectangle, scale=0.5]
\tikzstyle{bluetext}=[fill=none, draw=none, shape=rectangle, text={rgb,255: red,4; green,51; blue,255}]
\tikzstyle{reg}=[draw, fill=white, rounded rectangle, rounded rectangle left arc=none, minimum height=1em, minimum width=1em, node font={\scriptsize}]
\tikzstyle{coreg}=[draw, fill=white, rounded rectangle, rounded rectangle right arc=none, minimum height=1em, minimum width=1em, node font={\scriptsize}]
\tikzstyle{turquoisenode}=[fill={rgb,255: red,115; green,255; blue,239}, draw={rgb,255: red,115; green,255; blue,239}, shape=circle]
\tikzstyle{resistor}=[R]
\tikzstyle{inductor}=[L]
\tikzstyle{capacitor}=[C]
\tikzstyle{voltage-source}=[american voltage source]
\tikzstyle{current-source}=[american current source]

\tikzstyle{edge}=[-, draw=black]
\tikzstyle{diredge}=[->, draw=black]
\tikzstyle{dashed edge}=[-, dashed, dash pattern=on 1pt off 1.5pt, draw=black]
\tikzstyle{dirdash}=[->, dashed, dash pattern=on 2pt off 0.5pt, draw=black]
\tikzstyle{mapsto}=[{|->}, draw=black]
\tikzstyle{gray diredge}=[draw={rgb,255: red,221; green,221; blue,221}, ->]
\tikzstyle{dark grey dirdash}=[->, dashed, dash pattern=on 2pt off 0.5pt, draw={rgb,255: red,81; green,81; blue,81}]
\tikzstyle{doubedge}=[-, draw=black, double=none, double distance=3pt, inner sep=0pt, thick]
\tikzstyle{thedge}=[-, line width=1.5pt, draw=black]
\tikzstyle{gray dashed}=[-, dashed, dash pattern=on 1pt off 1.5pt, draw={rgb,255: red,128; green,128; blue,128}]
\tikzstyle{rededge}=[-, draw=red]
\tikzstyle{gray edge}=[-, draw={rgb,255: red,128; green,128; blue,128}]
\tikzstyle{blthedge}=[-, thick, draw={rgb,255: red,4; green,51; blue,255}]
\tikzstyle{blthdash}=[-, dashed, dash pattern=on 1pt off 1.5pt, thick, draw={rgb,255: red,4; green,51; blue,255}]
\tikzstyle{dirrededge}=[draw=red, ->]

\tikzstyle{object}=[inner sep=0mm, shape=circle, fill=none]
\tikzstyle{bullet}=[fill=black, draw=black, shape=circle, scale=0.3]
\tikzstyle{circ}=[fill=white, draw=black, shape=circle, scale=0.3]
\tikzstyle{objectbox}=[inner sep=3pt, shape=rectangle, fill=none]
\tikzstyle{bdyobjectbox}=[fill=white, draw=black, shape=rectangle]

\tikzstyle{morphism}=[->, draw=black]
\tikzstyle{dash morphism}=[->, dashed, dash pattern=on 1.5pt off 1pt, draw=black]
\tikzstyle{mapsto}=[{|->}, draw=black]
\tikzstyle{nat transf}=[-implies, double, double distance=3pt, thick]
\tikzstyle{gray nat transf}=[-implies, draw=gray, double, double distance=3pt, thick]
\tikzstyle{equality}=[-, double, double distance=3pt]
\tikzstyle{squig morphism}=[rightsquigarrow, draw=black]
\tikzstyle{hookarrow}=[right hook->, draw=black]
\tikzstyle{hookarrowmirror}=[left hook->, draw=black]

\graphicspath{ {figures/} }

\title{A Categorical Model for Retrosynthetic Reaction Analysis}
\author{
Ella Gale \footnote{University of Bristol, UK \texttt{ella.gale@bristol.ac.uk}} \and
Leo Lobski \footremember{ucl}{University College London, UK \texttt{leo.lobski.21@ucl.ac.uk f.zanasi@ucl.ac.uk}} \and
Fabio Zanasi \footrecall{ucl} \footnote{University of Bologna, Italy}
}
\date{}

\begin{document}
\maketitle

\begin{abstract}
We introduce a mathematical framework for retrosynthetic analysis, an important research method in synthetic chemistry. Our approach represents molecules and their interaction using string diagrams in layered props -- a recently introduced categorical model for partial explanations in scientific reasoning. Such principled approach allows one to model features currently not available in automated retrosynthesis tools, such as chirality, reaction environment and protection-deprotection steps.
\end{abstract}

\section{Introduction}

A chemical reaction can be understood as a rule which tells us what the outcome molecules (or molecule-like objects, such as ions) are when several molecules are put together. If, moreover, the reaction records the precise proportions of the molecules as well as the conditions for the reaction to take place (temperature, pressure, concentration, presence of a solvent etc.), it can be seen as a precise scientific prediction, whose truth or falsity can be tested in a lab, making the reaction reproducible. Producing complicated molecules, as required e.g.~by the pharmaceutical industry, requires, in general, a chain of several consecutive reactions in precisely specified conditions. The general task of synthetic chemistry is to come up with reproducible reaction chains to generate previously unknown molecules (with some desired properties)~\cite{organic-synthesis}. Successfully achieving a given synthetic task requires both understanding of the chemical mechanisms and the empirical knowledge of existing reactions. Both of these are increasingly supported by computational methods~\cite{strieth2020machine}: rule-based and dynamical models are used to suggest potential reaction mechanisms, while database search is used to look for existing reactions that would apply in the context of interest~\cite{compuuter-aided2022}. The key desiderata for such tools are tunability and specificity. Tunability endows a synthetic chemist with tools to specify a set of goals (e.g.~adding or removing a functional group\footnote{Part of a molecule that is known to be responsible for certain chemical function.}), while by specificity we mean maximising yield and minimising side products.

In this paper, we focus on the area of synthetic chemistry known as {\em retrosynthesis}~\cite{corey1988robert,compuuter-aided2022,warren1991designing}. While reaction prediction asks what reactions will occur and what outcomes will be obtained when some molecules are allowed to interact, retrosynthesis goes backwards: it starts with a target molecule that we wish to produce, and it proceeds in the ``reverse'' direction by asking what potential reactants would produce the target molecule. While many automated tools for retrosynthesis exist (see e.g.~\cite{route-designer,Coley2017-similarity,coley2018machine,Lin2020,Chen2021,ucak2022,dong2022deep}), there is no uniform mathematical framework in which the suggested algorithms could be analysed, compared or combined. The primary contribution of this paper is to provide such a framework. By formalising the methodology at this level of mathematical generality, we are able to provide insights into how to incorporate features that the current automated retrosynthesis tools lack: these include modelling chirality, the reaction environment, and the protection-deprotection steps (see for example~\cite{filice2010recent}), which are all highly relevant to practical applications. Our formalism, therefore, paves the way for new automated retrosynthesis tools, accounting for the aforementioned features.

Mathematically, our approach is phrased in the algebraic formalism of {\em string diagrams}~\cite{piedeleu-zanasi}, and most specifically uses {\em layered props}. Layered props were originally introduced, in~\cite{lobski-zanasi}, as models for systems that have several interdependent levels of description. In the context of chemistry, the description levels play a threefold role: first, each level represents a reaction environment, second, the morphisms in different levels are taking care of different synthetic tasks, and third, the rules that are available in a given level reflect the structure that is deemed relevant for the next retrosynthetic step. The latter can be seen as a kind of coarse-graining, where by deliberately restricting to a subset of all available information, we reveal some essential features about the system. Additionally, organising retrosynthetic rules into levels allows us to include conditions that certain parts of a molecule are to be kept intact. While the presentation here is self-contained and, in particular, does not assume a background on layered props, we emphasise that our approach is principled in the sense that many choices we make are suggested by this more general framework. We point such choices out when we feel the intuition that comes from layered props is helpful for understanding the formalism presented in the present work.

The rest of the paper is structured as follows. In Section~\ref{sec:chem-background}, we give a brief overview of the methodology of retrosynthetic analysis, as well as of the existing tools for automating it. Section~\ref{sec:layered-props} recalls the conceptual and mathematical ideas behind layered props. The entirety of Section~\ref{sec:chem-graphs} is devoted to constructing the labelled graphs that we use to represent molecular entities: these will be the objects of the monoidal categories we introduce in Sections~\ref{sec:reactions} and~\ref{sec:retrosynthesis}. Section~\ref{sec:disc-rules} formalises retrosynthetic disconnection rules, while Section~\ref{sec:reactions} formalises reactions. The culmination of the paper is the layered prop defined in Section~\ref{sec:retrosynthesis}, where we also describe how to reason about retrosynthesis within it. In Section~\ref{sec:discussion} we sketch the prospects of future work.

This work extends the conference paper {\em A Categorical Approach to Synthetic Chemistry} (ICTAC 2023) with new material in Sections~\ref{sec:reactions} and~\ref{sec:retrosynthesis}. We have included more detailed constructions, proofs of the results, as well as many examples to illustrate the definitions. The other sections are essentially unchanged from the conference version.

\section{Retrosynthetic Analysis}\label{sec:chem-background}

Retrosynthetic analysis starts with a target molecule we wish to produce but do not know how. The aim is to ``reduce'' the target molecule to known (commercially available) outcome molecules in such a way that when the outcome molecules react, the target molecule is obtained as a product. This is done by (formally) partitioning the target molecule into functional parts referred to as {\em synthons}, and finding actually existing molecules that are chemically equivalent to the synthons; these are referred to as {\em synthetic equivalents}~\cite{logic-chemical,organic-synthesis,organic-chemistry}. If no synthetic equivalents can be found that actually exist, the partitioning step can be repeated, this time using the synthetic equivalents themselves as the target molecules, and the process can continue until either known molecules are found, or a maximum number of steps is reached and the search is stopped. Note that the synthons themselves do not refer to any molecule as such, but are rather a convenient formal notation for parts of a molecule. For this reason, passing from synthons to synthetic equivalents is a non-trivial step involving intelligent guesswork and chemical know-how of how the synthons {\em would} react if they were independent chemical entities.
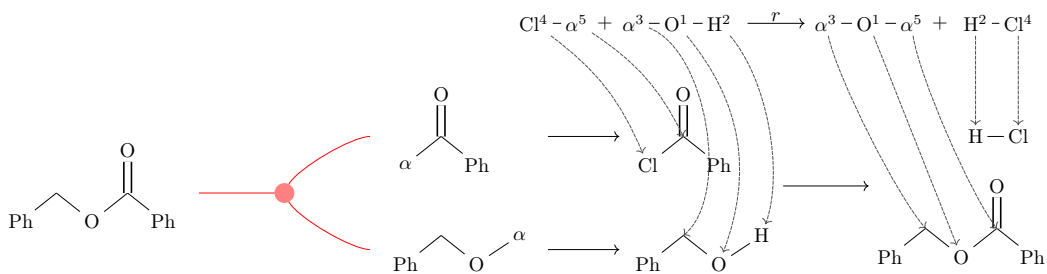
\begin{figure}
  \centering
    \scalebox{.75}{\begin{tikzpicture}
	\begin{pgfonlayer}{nodelayer}
		\node [style=textbox] (0) at (-14.5, -1) {Ph};
		\node [style=textbox] (1) at (-17, -1) {O};
		\node [style=textbox] (2) at (-19.5, -1) {Ph};
		\node [style=textbox] (3) at (-15.75, 1.5) {O};
		\node [style=none] (5) at (-18.25, 0) {};
		\node [style=none] (6) at (-15.75, 0) {};
		\node [style=none] (7) at (-13.25, 0) {};
		\node [style=none] (8) at (-7.25, 2) {};
		\node [style=rednode] (9) at (-10.25, 0) {};
		\node [style=none] (10) at (-7.25, -2) {};
		\node [style=textbox] (11) at (-3.5, 1) {Ph};
		\node [style=textbox] (12) at (-3.5, -2.5) {O};
		\node [style=textbox] (13) at (-6, -2.5) {Ph};
		\node [style=textbox] (14) at (-4.75, 3.5) {O};
		\node [style=none] (15) at (-4.75, -1.5) {};
		\node [style=none] (16) at (-4.75, 2) {};
		\node [style=textbox] (17) at (-2, -1.5) {$\alpha$};
		\node [style=textbox] (18) at (-6, 1) {$\alpha$};
		\node [style=textbox] (19) at (5, 1) {Ph};
		\node [style=textbox] (20) at (5, -2.5) {O};
		\node [style=textbox] (21) at (2.5, -2.5) {Ph};
		\node [style=textbox] (22) at (3.75, 3.5) {O};
		\node [style=none] (23) at (3.75, -1.5) {};
		\node [style=none] (24) at (3.75, 2) {};
		\node [style=textbox] (25) at (6.5, -1.5) {H};
		\node [style=textbox] (26) at (2.5, 1) {Cl};
		\node [style=none] (27) at (-1, 2) {};
		\node [style=none] (28) at (1.5, 2) {};
		\node [style=none] (29) at (-1, -2) {};
		\node [style=none] (30) at (1.5, -2) {};
		\node [style=textbox] (31) at (3.5, 6) {O$^1$};
		\node [style=textbox] (32) at (2, 6) {$\alpha^3$};
		\node [style=textbox] (33) at (5, 6) {H$^2$};
		\node [style=none] (34) at (1, 6) {$+$};
		\node [style=textbox] (35) at (-1.5, 6) {Cl$^4$};
		\node [style=textbox] (36) at (0, 6) {$\alpha^5$};
		\node [style=none] (37) at (6, 6) {};
		\node [style=none] (38) at (8, 6) {};
		\node [style=textbox] (39) at (10.25, 6) {O$^1$};
		\node [style=textbox] (40) at (11.75, 6) {$\alpha^5$};
		\node [style=textbox] (41) at (8.75, 6) {$\alpha^3$};
		\node [style=none] (42) at (12.75, 6) {$+$};
		\node [style=textbox] (43) at (15.5, 6) {Cl$^4$};
		\node [style=textbox] (44) at (14, 6) {H$^2$};
		\node [style=none] (45) at (7, 6.25) {$r$};
		\node [style=textbox] (46) at (16, -2.25) {Ph};
		\node [style=textbox] (47) at (13.5, -2.25) {O};
		\node [style=textbox] (48) at (11, -2.25) {Ph};
		\node [style=textbox] (49) at (14.75, 0.25) {O};
		\node [style=none] (50) at (12.25, -1.25) {};
		\node [style=none] (51) at (14.75, -1.25) {};
		\node [style=none] (52) at (7.25, 0.25) {};
		\node [style=none] (53) at (10.25, 0.25) {};
		\node [style=textbox] (54) at (15.5, 2) {Cl};
		\node [style=textbox] (55) at (14, 2) {H};
	\end{pgfonlayer}
	\begin{pgfonlayer}{edgelayer}
		\draw [style=edge] (2) to (5.center);
		\draw [style=edge] (5.center) to (1);
		\draw [style=edge] (1) to (6.center);
		\draw [style=edge] (6.center) to (0);
		\draw [style=doubedge] (6.center) to (3);
		\draw [style=rededge, bend right, looseness=0.50] (9) to (10.center);
		\draw [style=rededge] (7.center) to (9);
		\draw [style=rededge, bend left, looseness=0.50] (9) to (8.center);
		\draw [style=edge] (13) to (15.center);
		\draw [style=edge] (15.center) to (12);
		\draw [style=edge] (16.center) to (11);
		\draw [style=doubedge] (16.center) to (14);
		\draw [style=edge] (12) to (17);
		\draw [style=edge] (18) to (16.center);
		\draw [style=edge] (21) to (23.center);
		\draw [style=edge] (23.center) to (20);
		\draw [style=edge] (24.center) to (19);
		\draw [style=doubedge] (24.center) to (22);
		\draw [style=edge] (20) to (25);
		\draw [style=edge] (26) to (24.center);
		\draw [style=diredge] (27.center) to (28.center);
		\draw [style=diredge] (29.center) to (30.center);
		\draw [style=edge] (31) to (32);
		\draw [style=edge] (33) to (31);
		\draw [style=edge] (35) to (36);
		\draw [style=diredge] (37.center) to (38.center);
		\draw [style=edge] (39) to (40);
		\draw [style=edge] (41) to (39);
		\draw [style=edge] (43) to (44);
		\draw [style=edge] (48) to (50.center);
		\draw [style=edge] (50.center) to (47);
		\draw [style=edge] (47) to (51.center);
		\draw [style=edge] (51.center) to (46);
		\draw [style=doubedge] (51.center) to (49);
		\draw [style=diredge] (52.center) to (53.center);
		\draw [style=dark grey dirdash, bend left=15] (35) to (26);
		\draw [style=dark grey dirdash, bend left=15] (36) to (24.center);
		\draw [style=dark grey dirdash, bend left, looseness=0.75] (33) to (25);
		\draw [style=dark grey dirdash, bend left] (31) to (20);
		\draw [style=dark grey dirdash] (39) to (47);
		\draw [style=dark grey dirdash, bend right=15, looseness=0.50] (40) to (51.center);
		\draw [style=edge] (54) to (55);
		\draw [style=dark grey dirdash] (43) to (54);
		\draw [style=dark grey dirdash] (44) to (55);
		\draw [style=dark grey dirdash, bend left=60, looseness=0.75] (32) to (23.center);
		\draw [style=dark grey dirdash, bend right=15, looseness=0.50] (41) to (50.center);
	\end{pgfonlayer}
\end{tikzpicture}}
  \caption{A retrosynthetic sequence \label{fig:clayden}}
\end{figure}

Clayden, Warren and Greeves~\cite{organic-chemistry} give the example in Figure~\ref{fig:clayden} when introducing retrosynthesis. Here the molecule on the left-hand side is the target, the resulting two parts with the symbol $\alpha$ are the synthons. We use the symbol $\alpha$ to indicate where the cut has been made, and hence which atoms have unpaired electrons. Replacing the symbols $\alpha$ in the synthons with $\mathtt{Cl}$ and $\mathtt{H}$, we obtain the candidate synthetic equivalents shown one step further to the right. Assuming existence of the reaction scheme $r$ shown at the top, it can be shown that there is a reaction starting from the synthetic equivalents and ending with the target. This is the simplest possible instance of a retrosynthetic sequence. In general, the interesting sequences are much longer, and, importantly, contain information under what conditions the reactions will take place.

\paragraph{Existing tools.}
Many tools for automatic retrosynthesis have been successfully developed starting from the 1960s~\cite{route-designer,Coley2017-similarity,Lin2020,Chen2021,ucak2022}. They can be divided into two classes~\cite{compuuter-aided2022}: {\em template-based}~\cite{fortunato2020data,yan2022retrocomposer} and {\em template-free}~\cite{Lin2020,somnath2020learning}. Template-based tools contain a rule database (the {\em template}), which is either manually encoded or automatically extracted. Given a molecule represented as a graph, the model checks whether any rules are applicable to it by going through the database and comparing the conditions of applying the rule to the subgraphs of the molecule~\cite{compuuter-aided2022}. Choosing the order in which the rules from the template and the subgraphs are tried are part of the model design. Template-free tools, on the other hand, are data-driven and treat the retrosynthetic rule application as a translation between graphs or their representations as strings: the suggested transforms are based on learning from known transforms, avoiding the need for a database of rules~\cite{compuuter-aided2022,ucak2022}.

While successful retrosynthesic sequences have been predicted by the computational retrosynthesis tools, they lack a rigorous mathematical foundation, which makes them difficult to compare, combine or modify. Other common drawbacks of the existing approaches include not including the reaction conditions or all cases of chirality as part of the reaction template~\cite{compuuter-aided2022,Lin2020}, as well as the fact that the existing models are unlikely to suggest protection-deprotection steps. Additionally, the template-free tools based on machine learning techniques sometimes produce output that does not correspond to molecules in any obvious way, and tend to reproduce the biases present in the literature or a data set~\cite{compuuter-aided2022}.

For successful prediction, the reaction conditions are, of course, crucial. These include such factors as temperature and pressure, the presence of a solvent (a compound which takes part in the reaction and whose supply is essentially unbounded), the presence of a reagent (a compound without which the reaction would not occur, but which is not the main focus or the target), as well as the presence of a catalyst (a compound which increases the rate at which the reaction occurs, but is itself unaltered by the reaction). The above factors can change the outcome of a reaction dramatically~\cite{matwijczuk2017effect,cook2007determination}. There have indeed been several attempts to include reaction conditions into the forward reaction prediction models~\cite{marcou2015,gao2018,walker2019,maser2021}. However, the search space in retrosynthesis is already so large that adding another search criterion should be done with caution. A major challenge for predicting reaction conditions is that they tend to be reported incompletely or inconsistently in the reaction databases~\cite{coley2017-outcomes}.

Chirality (mirror-image asymmetery) of a molecule can alter its chemical and physiological properties, and hence constitutes a major part of chemical information pertaining to a molecule. While template-based methods have been able to successfully suggest reactions involving chirality (e.g.~\cite{Coley2017-similarity}), the template-free models have difficulties handling it~\cite{Lin2020}. This further emphasises usefulness of a framework which is able to handle both template-based and template-free models.

The protection-deprotection steps are needed when more than one functional group of a molecule $A$ would react with a molecule $B$. To ensure the desired reaction, the undesired functional group of $A$ is first ``protected'' by adding a molecule $X$, which guarantees that the reaction product will react with $B$ in the required way. Finally, the protected group is ``deprotected'', producing the desired outcome of $B$ reacting with the correct functional group of $A$. So, instead of having a direct reaction $A+B\rightarrow C$ (which would not happen, or would happen imperfectly, due to a ``competing'' functional group), the reaction chain is:
$$\text{(1) } A + X \rightarrow A'\text{ (protection)},\quad\text{(2) } A' + B \rightarrow C',\quad\text{(3) } C' + Y \rightarrow C\text{ (deprotection).}$$
The trouble with the protection-deprotection steps is that they temporarily make the molecule larger, which means that an algorithm whose aim is to make a molecule smaller will not suggest them.

\section{Layered Props}\label{sec:layered-props}

{\em Layered props} were introduced in~\cite{lobski-zanasi} as categorical models for diagrammatic reasoning about systems with several levels of description. They have been employed to account for partial explanations and semantic analysis in the context of electrical circuit theory, chemistry, and concurrency. Formally, a layered prop is essentially a functor $\Omega:P\rightarrow\StrMon$ from a poset $P$ to the category of strict monoidal categories, together with a right adjoint for each monoidal functor in the image of $\Omega$. Given $\omega\in P$, we denote a morphism $\sigma:a\rightarrow b$ in $\Omega(\omega)$ by the box on the
\begin{wrapfigure}[3]{r}{0.2\textwidth}
    \vspace{-8pt}
    \centering
    \scalebox{.8}{\begin{tikzpicture}
	\begin{pgfonlayer}{nodelayer}
		\node [style=none] (0) at (-2, 1) {};
		\node [style=none] (1) at (-2, -1) {};
		\node [style=none] (2) at (2, 1) {};
		\node [style=none] (3) at (2, -1) {};
		\node [style=none] (4) at (-2, 0) {};
		\node [style=none] (5) at (2, 0) {};
		\node [style=redbox] (6) at (0, 0) {$\sigma$};
		\node [style=none] (7) at (-1, 0.25) {$a$};
		\node [style=none] (8) at (-3, 0) {$\omega$};
		\node [style=none] (9) at (1, 0.25) {$b$};
		\node [style=none] (10) at (3, 0) {$\omega$};
	\end{pgfonlayer}
	\begin{pgfonlayer}{edgelayer}
		\draw [style=rededge] (4.center) to (6);
		\draw [style=rededge] (6) to (5.center);
		\draw [style=edge] (0.center) to (1.center);
		\draw [style=edge] (1.center) to (3.center);
		\draw [style=edge] (3.center) to (2.center);
		\draw [style=edge] (2.center) to (0.center);
	\end{pgfonlayer}
\end{tikzpicture}}
\end{wrapfigure}
right. We think of $\sigma$ as a {\em process} with an input $a$ and an output $b$ happening in the {\em context} $\omega$. Note, however, that these diagrams are not merely a convenient piece of notation that capture our intuition: they are a completely formal syntax of string diagrams, describing morphisms in a certain subcategory of pointed profunctors~\cite{lobski-zanasi}.

The monoidal categories in the image of $\Omega$ are thought of as languages describing the same system at different levels of granularity, and the functors are seen as translations between the languages. Given $\omega\leq\tau$ in $P$, let us write $f\coloneqq\Omega(\omega\leq\tau)$. Then, for each $a\in\Omega(\omega)$ we have the
\begin{wrapfigure}[7]{r}{0.2\textwidth}
    \vspace{-14pt}
    \centering
    \scalebox{.8}{\begin{tikzpicture}
	\begin{pgfonlayer}{nodelayer}
		\node [style=none] (0) at (-3, 0) {$\omega$};
		\node [style=none] (1) at (-2, 1) {};
		\node [style=none] (2) at (-2, -1) {};
		\node [style=none] (3) at (0, 1) {};
		\node [style=none] (4) at (0, -1) {};
		\node [style=none] (5) at (2, 1) {};
		\node [style=none] (6) at (2, -1) {};
		\node [style=none] (7) at (3, 0) {$\tau$};
		\node [style=none] (8) at (0, 1.5) {$\refine_f$};
		\node [style=none] (9) at (-2, 0) {};
		\node [style=none] (10) at (2, 0) {};
		\node [style=none] (11) at (-1, 0.25) {$a$};
		\node [style=none] (12) at (1, 0.25) {$fa$};
		\node [style=none] (13) at (0, 2) {};
	\end{pgfonlayer}
	\begin{pgfonlayer}{edgelayer}
		\draw [style=rededge] (9.center) to (10.center);
		\draw [style=edge] (3.center) to (4.center);
		\draw [style=edge] (1.center) to (5.center);
		\draw [style=edge] (1.center) to (2.center);
		\draw [style=edge] (2.center) to (6.center);
		\draw [style=edge] (6.center) to (5.center);
	\end{pgfonlayer}
\end{tikzpicture}} \\
    \vspace{6pt}
    \scalebox{.8}{\begin{tikzpicture}
	\begin{pgfonlayer}{nodelayer}
		\node [style=none] (0) at (3, 0) {$\omega$};
		\node [style=none] (1) at (2, 1) {};
		\node [style=none] (2) at (2, -1) {};
		\node [style=none] (3) at (0, 1) {};
		\node [style=none] (4) at (0, -1) {};
		\node [style=none] (5) at (-2, 1) {};
		\node [style=none] (6) at (-2, -1) {};
		\node [style=none] (7) at (-3, 0) {$\tau$};
		\node [style=none] (8) at (0, 1.5) {$\coarsen_f$};
		\node [style=none] (9) at (2, 0) {};
		\node [style=none] (10) at (-2, 0) {};
		\node [style=none] (11) at (1, 0.25) {$a$};
		\node [style=none] (12) at (-1, 0.25) {$fa$};
		\node [style=none] (13) at (0, 2) {};
	\end{pgfonlayer}
	\begin{pgfonlayer}{edgelayer}
		\draw [style=rededge] (9.center) to (10.center);
		\draw [style=edge] (3.center) to (4.center);
		\draw [style=edge] (5.center) to (1.center);
		\draw [style=edge] (1.center) to (2.center);
		\draw [style=edge] (2.center) to (6.center);
		\draw [style=edge] (6.center) to (5.center);
	\end{pgfonlayer}
\end{tikzpicture}}
\end{wrapfigure}
morphisms drawn on the right. The reason for having morphisms in both directions is that we want to be able to ``undo'' the action of a translation while preserving a linear reasoning flow. The two morphisms will not, in general, be inverse to each other: rather, they form an adjoint pair. This corresponds to the intuition that some information is gained by performing the translation, and that the translation in the reverse direction is our best guess, or an approximation, not a one-to-one correspondence.

There are two ways to compose morphisms in parallel in a layered prop: internally within a monoidal category $\Omega(\omega)$ using its own monoidal product (composition inside a context), and externally using the Cartesian monoidal structure of $\StrMon$ (doing several processes in different contexts in parallel). We represent the latter by stacking the boxes on top of each other. Additional morphisms of a layered prop ensure that the internal and the external monoidal structures interact in a coherent way. Finally, a layered prop comes with ``deduction rules'' (2-cells) which allow transforming one process into another one. We refer the reader to~\cite{lobski-zanasi} for the details.

In this work, the processes in context will be the retrosynthetic disconnection rules (Section~\ref{sec:disc-rules}) and the chemical reactions (Section~\ref{sec:reactions}). The context describes the reaction environment as well as the level of granularity at which the synthesis is happening (i.e.~what kinds of disconnection rules are available). The objects in the monoidal categories are given by molecular entities and their parts: this is the subject of the next section.

\section{Chemical Graphs}\label{sec:chem-graphs}
We define a chemical graph as a labelled graph whose edge labels indicate the bond type (covalent, ionic), and whose vertex labels are either atoms, charges or unpaired electrons (Definitions~\ref{def:prechemgraph} and~\ref{def:chemgraph}). In order to account for chirality, we add spatial information to chemical graphs, making it an {\em oriented (pre-)chemical graph} (Definition~\ref{def:orient-molpart}).

Oriented chemical graphs form the objects of the layered props we suggest as a framework for synthetic chemistry. The morphisms of these layered prop correspond to retrosynthetic disconnection rules and chemical reactions; this is the topic of the next two sections.

Let us define the set of {\em atoms} as containing the symbol for each main-group element of the periodic table: $\At\coloneqq\{H,C,O,P,\dots\}$. Define the function $\mathbf v:\At\sqcup\{+,-,\alpha\}\rightarrow\N$ as taking each element symbol to its valence\footnote{This is a bit of a naive model, as valence is, in general, context-sensitive. We leave accounting for this to future work.}, and define $\mathbf v(-)=\mathbf v(+)=\mathbf v(\alpha)=1$, where $+$ and $-$ stand for positive and negative charge, while $\alpha$ denotes an unpaired electron. Let $\Lab\coloneqq\{0,1,2,3,4,\ib\}$ denote the set of {\em edge labels}, where the integers stand for a covalent bond, and $\ib$ for an ionic bond. We further define maps $\cov,\ion:\Lab\rightarrow\N$: for $\cov$, assign to each edge label $0$, $1$, $2$, $3$, and $4$ the corresponding natural number and let $\ib\mapsto 0$, while for $\ion$, let $0,1,2,3,4\mapsto 0$ and $\ib\mapsto 1$. Finally, let us fix a countable set $\VS$ of {\em vertex names}; we usually denote the elements of $\VS$ by lowercase Latin letters $u,v,w,\dots$.

\begin{definition}[Pre-chemical graph]\label{def:prechemgraph}
A {\em pre-chemical graph} is a triple $(V,\tau,m)$, where $V\sse\VS$ is a finite set of {\em vertices}, $\tau:V\rightarrow\At\sqcup\{+,-,\alpha\}$ is a {\em vertex labelling function}, and $m:V\times V\rightarrow\Lab$ is an {\em edge labelling function} satisfying $m(v,v)=0$ and $m(v,w)=m(w,v)$ for all $v,w\in V$.
\end{definition}
Thus, a pre-chemical graph is irreflexive (we interpet the edge label $0$ as no edge) and symmetric, and each of its vertices is labelled with an atom, a charge or a placeholder variable $\alpha$. Given a pre-chemical graph $A$, we write $(V_A,\tau_A,m_A)$ for its vertex set and the labelling functions. Further, we define the following special subsets of vertices:
\begin{itemize}
\item {\em $\alpha$-vertices}, whose label is the placeholder variable: $\alpha(A)\coloneqq\tau^{-1}(\alpha)$,
\item {\em chemical vertices}, whose label is either an atom or a charge: $\Chem A\coloneqq V_A\setminus\alpha(A)$,
\item {\em charged vertices}, whose label is a charge: $\Crg A\coloneqq\tau^{-1}(\{+,-\})$,
\item {\em neutral vertices}, whose label is an atom or the placeholder variable: $\Neu A\coloneqq V_A\setminus\Crg A$.
\end{itemize}
For a pre-chemical graph $A$, its {\em net charge} is the integer $|\tau_A^{-1}(+)|-|\tau_A^{-1}(-)|$ given by subtracting the number of vertices with label $-$ from the number of $+$-labelled vertices.

Note that the collection of pre-chemical graphs has a partial monoid structure given by the disjoint union of labelled graphs, provided that the vertex sets are disjoint.
\begin{definition}[Chemical graph]\label{def:chemgraph}
A {\em chemical graph} $(V,\tau,m)$ is a pre-chemical graph satisfying the following additional conditions:
\begin{enumerate}
\item for all $v\in V$, we have $\sum_{u\in V}\cov\left(m(u,v)\right)=\mathbf v\tau(v)$,\label{cgraph:val}
\item for all $v,w\in V$ with $\tau(v)=\alpha$ and $m(v,w)=1$, we have $\tau(w)\in\At\sqcup\{-\}$,\label{cgraph:alpha}
\item if $v,w\in V$ such that $\tau(v)\in\{+,-\}$ and $m(v,w)=1$, then $\tau(w)\in\At\sqcup\{\alpha\}$,\label{cgraph:charge}
\item if $m(v,w)=\ib$, then 
\begin{enumerate}
\item $\tau(v),\tau(w)\in\{+,-\}$ and $\tau(v)\neq\tau(w)$,\label{cgraph:ion1}
\item for $a,b\in V$ with $m(v,a)=m(w,b)=1$, we have $\tau(a),\tau(b)\in\At$,\label{cgraph:ion2}
\item if for some $w'\in V$ we have $m(v,w')=\ib$, then $w=w'$.\label{cgraph:ion3}
\end{enumerate}
\end{enumerate}
\end{definition}
Condition~\ref{cgraph:val} says that the sum of each row or column in the adjacency matrix formed by the the integers $\cov\left(m(u,v)\right)$ gives the valence of the (label of) corresponding vertex. Conditions~\ref{cgraph:alpha} and~\ref{cgraph:charge} say that a vertex labelled by $\alpha$, $+$ or $-$ has to be connected to an atom, with the exception that the vertices labelled $\alpha$ and $-$ are allowed to be connected to each other instead of atoms. Finally, conditions~\ref{cgraph:ion1}-\ref{cgraph:ion3} say that an edge with label $\ib$ only connects vertices labelled with opposite charges ($+$ and $-$) that are themselves connected to atoms, such that each charge-labelled vertex is connected to at most one other such vertex.

A {\em synthon} is a chemical graph which is moreover connected. The collection of chemical graphs is, therefore, generated by the disjoint unions of synthons. A {\em molecular graph} is a chemical graph with no $\alpha$-vertices. A {\em molecular entity} is a connected molecular graph.

When drawing a chemical graph, we simply replace the vertices by their labels, unless the precise vertex names play a role. We adopt the usual chemical notation for $n$-ary bonds by drawing them as $n$ parallel lines. The ionic bonds are drawn as dashed lines.

\begin{example}\label{ex:partition}
We give examples of a synthon on the left, and two moleculear entities on the right: a molecule (ethenone) and an ion (carbonate anion).
\begin{center}
\begin{tabular}{c c c}
\scalebox{.85}{\begin{tikzpicture}
	\begin{pgfonlayer}{nodelayer}
		\node [style=textbox] (0) at (-0.75, 0) {C};
		\node [style=textbox] (1) at (-2, 1.25) {$\alpha$};
		\node [style=textbox] (2) at (0.75, 0) {C};
		\node [style=textbox] (3) at (-2, -1.25) {H};
		\node [style=textbox] (6) at (0.75, 1.5) {O};
		\node [style=textbox] (7) at (2, -1.25) {$\alpha$};
		\node [style=textbox] (8) at (2, 2.75) {$\alpha$};
	\end{pgfonlayer}
	\begin{pgfonlayer}{edgelayer}
		\draw [style=edge] (1) to (0);
		\draw [style=edge] (3) to (0);
		\draw [style=edge] (6) to (2);
		\draw [style=doubedge] (0) to (2);
		\draw [style=edge] (2) to (7);
		\draw [style=edge] (6) to (8);
	\end{pgfonlayer}
\end{tikzpicture}} \qquad & \scalebox{.85}{\begin{tikzpicture}
	\begin{pgfonlayer}{nodelayer}
		\node [style=textbox] (0) at (-1, 0) {C};
		\node [style=textbox] (1) at (-2.25, 1.25) {H};
		\node [style=textbox] (2) at (0.5, 0) {C};
		\node [style=textbox] (3) at (-2.25, -1.25) {H};
		\node [style=textbox] (6) at (2, 0) {O};
	\end{pgfonlayer}
	\begin{pgfonlayer}{edgelayer}
		\draw [style=edge] (1) to (0);
		\draw [style=edge] (3) to (0);
		\draw [style=doubedge] (0) to (2);
		\draw [style=doubedge] (2) to (6);
	\end{pgfonlayer}
\end{tikzpicture}} \qquad & \scalebox{.85}{\begin{tikzpicture}
	\begin{pgfonlayer}{nodelayer}
		\node [style=textbox] (0) at (0, 1.25) {C};
		\node [style=textbox] (1) at (0, 2.75) {O};
		\node [style=textbox] (2) at (-2, -1) {$-$};
		\node [style=textbox] (3) at (-1.25, 0) {O};
		\node [style=textbox] (4) at (1.25, 0) {O};
		\node [style=textbox] (5) at (2, -1) {$-$};
	\end{pgfonlayer}
	\begin{pgfonlayer}{edgelayer}
		\draw [style=edge] (0) to (4);
		\draw [style=edge] (3) to (0);
		\draw [style=edge] (2) to (3);
		\draw [style=edge] (4) to (5);
		\draw [style=doubedge] (1) to (0);
	\end{pgfonlayer}
\end{tikzpicture}}
\end{tabular}
\end{center}
\end{example}

\subsection{Chirality}
Next, we introduce (rudimentary) spatial information into (pre-)chemical graphs. The idea is to record for each triple of atoms whether they are on the same line or not, and similarly, for each quadruple of atoms whether they are in the same plane or not.
  \begin{definition}[Triangle relation]\label{def:plane-rel}
    Let $S$ be a set. We call a ternary relation $\mathcal P\sse S\times S\times S$ a {\em triangle relation} if the following hold for all elements $A$, $B$ and $C$ of $S$: (1) $ABB\notin\mathcal P$, and (2) if $\mathcal P(ABC)$ and $\mathfrak p(ABC)$ is any permutation of the three elements, then $\mathcal P(\mathfrak p(ABC))$.
  \end{definition}
  \begin{definition}[Tetrahedron relation]\label{def:tet-rel}
    Let $S$ be a set, and let $\mathcal P$ be a fixed triangle relation on $S$. We call a quaternary relation $\mathcal T\sse S\times S\times S\times S$ a {\em tetrahedron relation} if the following hold for all elements $A$, $B$, $C$ and $D$ of $S$: (1) if $\mathcal T(ABCD)$, then $\mathcal P(ABC)$, and (2) if $\mathcal T(ABCD)$ and $\mathfrak p(ABCD)$ is any even permutation of the four elements, then $\mathcal T(\mathfrak p(ABCD))$.
  \end{definition}
Unpacking the above definitions, a triangle relation is closed under the action of the symmetric
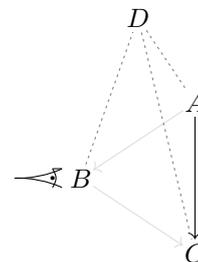
\begin{wrapfigure}[13]{r}{0.25\textwidth}
  \centering
    \vspace{-12pt}
    \begin{tikzpicture}
	\begin{pgfonlayer}{nodelayer}
		\node [style=textnode] (0) at (-3, 0) {$B$};
		\node [style=textnode] (1) at (0, 2) {$A$};
		\node [style=textnode] (2) at (0, -2) {$C$};
		\node [style=textnode] (3) at (-1.5, 4.25) {$D$};
		\node [style=none] (4) at (-4.75, 0) {};
		\node [style=none] (6) at (-3.75, -0.25) {};
		\node [style=none] (7) at (-3.5, -0.25) {};
		\node [style=none] (8) at (-3.5, 0.25) {};
		\node [style=none] (9) at (-3.75, 0) {\scalebox{.4}{$\bullet$}};
		\node [style=none] (10) at (-3.5, 0.25) {};
		\node [style=none] (11) at (-3.75, 0.25) {};
	\end{pgfonlayer}
	\begin{pgfonlayer}{edgelayer}
		\draw [style=gray diredge] (1) to (0);
		\draw [style=gray diredge] (0) to (2);
		\draw [style=diredge] (1) to (2);
		\draw [style=gray dashed] (0) to (3);
		\draw [style=gray dashed] (3) to (2);
		\draw [style=gray dashed] (3) to (1);
		\draw [style=edge, bend right=15, looseness=0.75] (4.center) to (8.center);
		\draw [style=edge, bend left=15, looseness=0.75] (4.center) to (7.center);
		\draw [style=edge] (11.center) to (10.center);
		\draw [style=edge, bend left, looseness=1.25] (11.center) to (6.center);
	\end{pgfonlayer}
\end{tikzpicture}
  \caption{Observer looking at the edge $AC$ from $B$ sees $D$ on their right. \label{fig:observer}}
\end{wrapfigure}
group $S_3$ such that any three elements it relates are pairwise distinct, and a tetrahedron relation is closed under the action of the alternating group $A_4$ such that if it relates some four elements, then the first three are related by some (fixed) triangle relation (this, inter alia, implies that any related elements are pairwise distinct, and their any $3$-element subset is related by the fixed triangle relation).

The intuition is that the triangle and tetrahedron relations capture the spatial relations of (not) being on the same line or plane: $\mathcal P(ABC)$ stands for $A$, $B$ and $C$ not being on the same line, that is, determining a triangle; similarly, $\mathcal T(ABCD)$ stands for $A$, $B$, $C$ and $D$ not being in the same plane, that is, determining a tetrahedron. The tetrahedron is moreover oriented: $\mathcal T(ABCD)$ does not, in general, imply $\mathcal T(DABC)$. We visualise $\mathcal T(ABCD)$ in Figure~\ref{fig:observer} by placing an ``observer'' at $B$ who is looking at the edge $AC$ such that $A$ is above $C$ for them. Then $D$ is on the right for this observer. Placing an observer in the same way in a situation where $\mathcal T(DABC)$ (which is equivalent to $\mathcal T(CBAD)$), they now see $D$ on their left.
\begin{remark}
We chose not to include the orientation of the triangle, which amounts to the choice of $S_3$ over $A_3$ in the definition of a triangle relation (Definition~\ref{def:plane-rel}). This is because we assume that our molecules float freely in space (e.g.~in a solution), so that there is no two-dimensional orientation.
\end{remark}

The following example demonstrates that the triangle and tetrahedron relations indeed capture triangles and tetrahedrons in the Euclidean setting.
\begin{example}
Let us define the triangle relation $\mathcal P$ on the $3$-dimensional Euclidean space $\R^3$ by letting $\mathcal P(abc)$ if and only if $(b-a)\times (c-a)\neq 0$, where $\times$ denotes the vector product. We then have $\mathcal P(abc)$ precisely when $c$ does not lie on the line determined by $a$ and $b$, that is, when the three points uniquely determine a plane in $\R^3$.

With respect to the above triangle relation, let us define the tetrahedron relation $\mathcal T$ by letting $\mathcal T(abcd)$ if and only if $\overline{(b-a)(c-a)(d-a)}>0$, where the bar denotes the scalar triple product. We then have $\mathcal T(abcd)$ precisely when the points $a$, $b$, $c$ and $d$ are vertices of a non-degenerate (a non-zero volume) tetrahedron in such a way that $d$ lies on that side of the plane determined by $a$, $b$ and $c$ to which the vector $(b-a)\times (c-a)$ points (see Figure~\ref{fig:observer}).
\end{example}

\begin{definition}[Oriented pre-chemical graph]\label{def:orient-molpart}
An {\em oriented pre-chemical graph} is a tuple $(V,\tau,m,\mathcal P,\mathcal T)$ where $(V,\tau,m)$ is a pre-chemical graph, $\mathcal P$ is a triangle relation on $V$ and $\mathcal T$ is a tetrahedron relation on $V$ with respect to $\mathcal P$, such that for all $a,b,c\in V$ we have: if $\mathcal P(abc)$, then $a,b,c\in\Neu V$, and moreover at most one of $a$, $b$ and $c$ is in $\alpha(V)$.
\end{definition}
An {\em oriented chemical graph} is an oriented pre-chemical graph, which is also a chemical graph.
\begin{definition}[Preservation and reflection of orientation]
Let $(M,\mathcal P_M,\mathcal T_M)$ and $(N,\mathcal P_N,\mathcal T_N)$ be oriented pre-chemical graphs, and let $f:M\rightarrow N$ be a labelled graph isomorphism. We say that $f$ {\em preserves orientation} (or is {\em orientation-preserving}) if for all vertices $A$, $B$, $C$ and $D$ of $M$ we have: (1) $\mathcal P_M(ABC)$ if and only if $\mathcal P_N(fA,fB,fC)$, and (2) $\mathcal T_M(ABCD)$ if and only if $\mathcal T_N(fA,fB,fC,fD)$.

Similarly, we say that $f$ {\em reflects orientation} (or is {\em orientation-reflecting}) if for all vertices $A$, $B$, $C$ and $D$ of $M$ we have: (1) $\mathcal P_M(ABC)$ if and only if $\mathcal P_N(fA,fB,fC)$, and (2) $\mathcal T_M(ABCD)$ if and only if $\mathcal T_N(fD,fA,fB,fC)$.
\end{definition}
\begin{definition}[Chirality]\label{def:chirality}
We say that two oriented pre-chemical graphs are {\em chiral} if there is an orientation-reflecting isomorphism, but no orientation-preserving isomorphism between them.
\end{definition}
\begin{example}\label{ex:2butanol}
Consider 2-butanol, whose molecular structure we draw in two different ways at the left of Figure~\ref{fig:example1}. Here we adopt the usual chemical convention for drawing spatial structure: a dashed wedge indicates that the bond points ``into the page'', and a solid wedge indicates that the bond points ``out of the page''. In this case, we choose to include the names of the vertices for some labels as superscripts. The spatial structure is formalised by defining the tetrahedron relation for the graph on the left-hand side as the closure under the action of $A_4$ of $\mathcal T(1234)$, and for the one on the right-hand side as (the closure of) $\mathcal T(4123)$. In both cases, the triangle relation is dictated by the tetrahedron relation, so that any three-element subset of $\{1,2,3,4\}$ is in the triangle relation. Now the identity map (on labelled graphs) reflects orientation. It is furthermore not hard to see that every isomorphism restricts to the identity on the vertices labelled with superscripts, so that there is no orientation-preserving isomorphism. Thus the two molecules are chiral according to Definition~\ref{def:chirality}.

By slightly modifying the structures, we obtain two configurations of isopentane, drawn at the right of Figure~\ref{fig:example1}. However, in this case we can find an orientation-preserving isomorphism (namely the one that swaps vertices $2$ and $4$), so that the molecules are not chiral.
\end{example}
\begin{figure}
  \centering
    \includegraphics[scale=0.23]{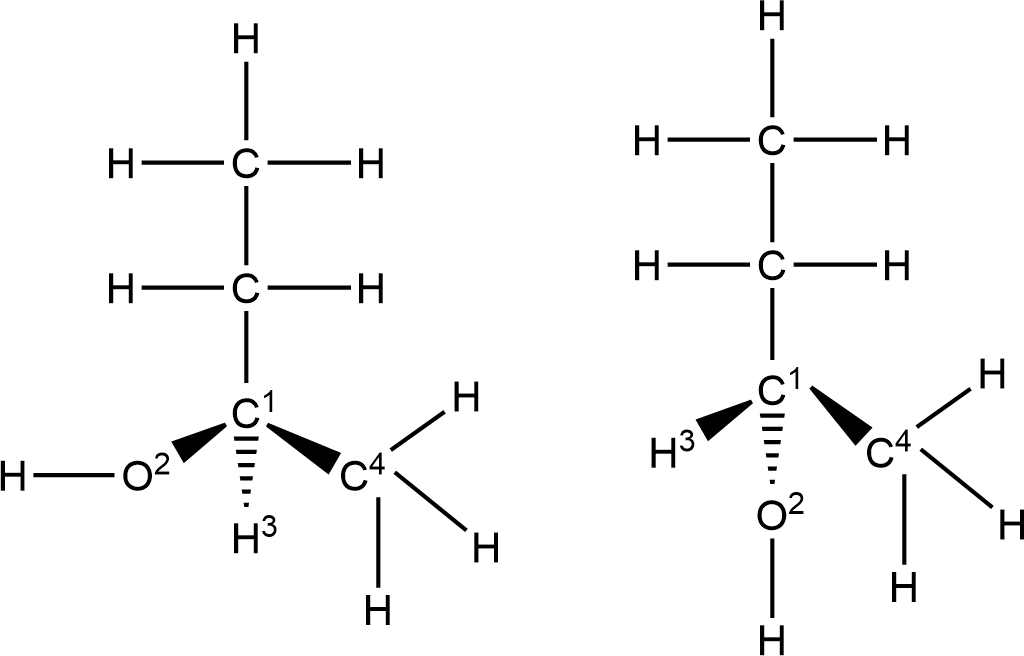} \qquad
    \includegraphics[scale=0.23]{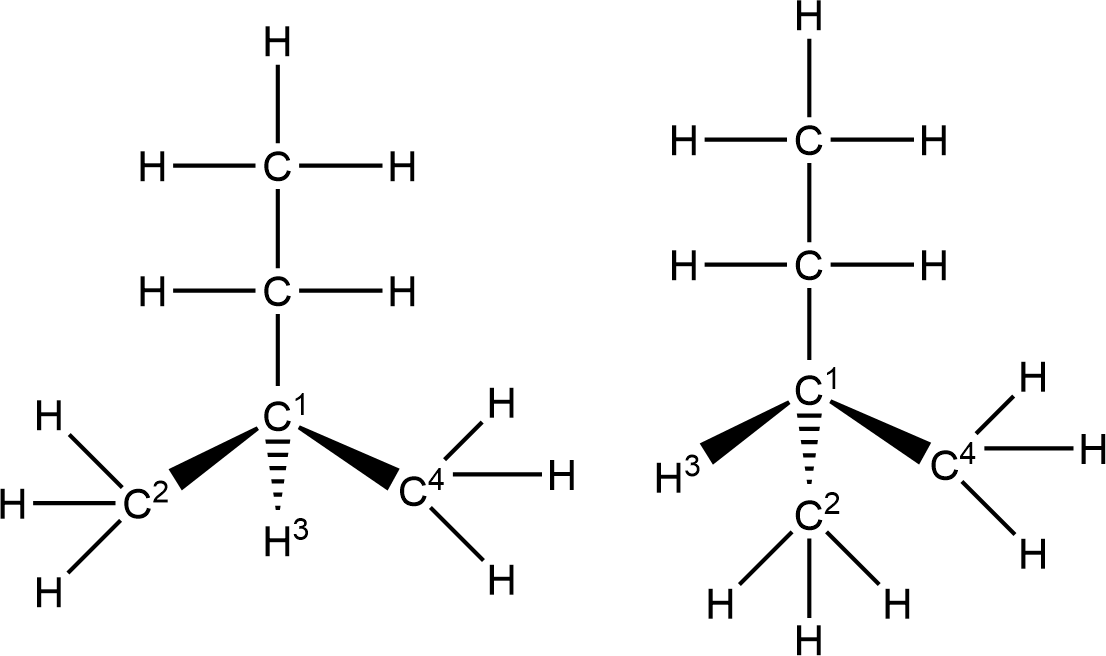}
    \caption{Left: two configurations of 2-butanol. Right: two configurations of isopentane. \label{fig:example1}}
\end{figure}
\begin{example}
Example~\ref{ex:2butanol} with 2-butanol demonstrated how to capture central chirality using Definition~\ref{def:chirality}. In this example, we consider 1,3-dichloroallene as an example of axial chirality. We draw two versions, as before:
\begin{center}
\includegraphics[scale=0.23]{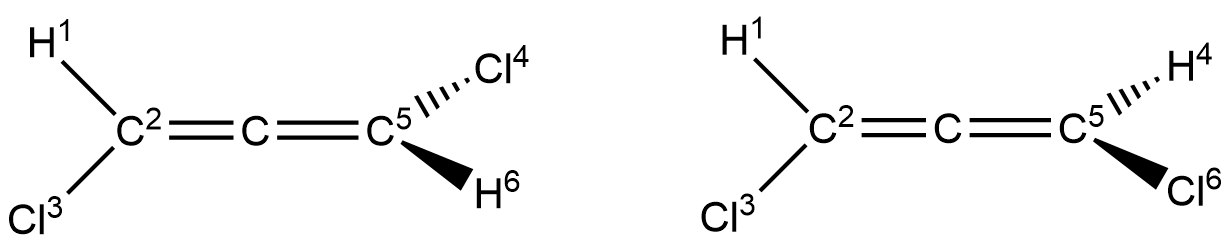}
\end{center}
The tetrahedron relation is generated by $\mathcal T(1234)$ and $\mathcal T(6123)$ for both molecules (note, however, that the vertices $4$ and $6$ have different labels). Now the isomorphism which swaps vertices $4$ and $6$ and is identity on all other vertices is orientation-reflecting, but not orientation-preserving. The only other isomorphism is $1\mapsto 4$, $2\mapsto 5$, $3\mapsto 6$, $4\mapsto 3$, $5\mapsto 2$, $6\mapsto 1$, which does not preserve orientation. Thus the two molecules are indeed chiral.
\end{example}

\section{Disconnection Rules}\label{sec:disc-rules}

The backbone of retrosynthetic analysis are the {\em disconnection rules} that partition the target molecule into smaller parts. Formally, a disconnection rule is a partial endofunction on the set of chemical graphs. We define three classes of disconnection rules, all of which have a clear chemical significance: {\em electron detachment} (Def.~\ref{def:edetach}), {\em ionic bond breaking} (Def.~\ref{def:iobond}) and {\em covalent bond breaking} (Def.~\ref{def:covbond}).

The above three classes of rules are chosen since they are used in the current retrosynthesis practice (e.g.~\cite{organic-synthesis,organic-chemistry}). However, once the reverse ``connection'' rules are added, we also conjecture that the rules are complete in the sense that every reaction (Def.~\ref{def:reaction-scheme}) can be decomposed into a sequence of disconnection rules. This conjecture comes with a caveat: the disconnection/connection rules cannot capture orientation, as they are local rewrite rules (acting on one edge at a time), while orientation contains global information about sets of vertices. For this reason, in this section we work with chemical graphs without an orientation.

\begin{definition}[Electron detachment]\label{def:edetach}
Let $u,v,a,b\in\VS$ be pairwise distinct vertex names. We define the {\em electron detachment} disconnection rule $E^{uv}_{ab}$ as follows:
\begin{itemize}
\item a chemical graph $A=(V,\tau,m)$ is in the domain of $E^{uv}_{ab}$ if (1) $u,v\in V$, (2) $a,b\notin V$, (3) $\tau(u)\in\At$, (4) $\tau(v)=\alpha$, and (5) $m(u,v)=1$,
\item the chemical graph $E^{uv}_{ab}(A)=(V\cup\{a,b\}, \tau^E, m^E)$ is defined by letting $\tau^E(a)=+$, $\tau^E(b)=-$ and letting $\tau^E$ agree with $\tau$ otherwise; further, define $m^E(u,v)=m^E(a,b)=0$, $m^E(u,a)=m^E(v,b)=1$ and let $m^E$ agree with $m$ otherwise.
\end{itemize}
\end{definition}
\begin{example}
The effect of the electron detachment is to detach an electron from a synthon, thus leaving it with a positive charge:
\begin{center}
\scalebox{.85}{\begin{tikzpicture}
	\begin{pgfonlayer}{nodelayer}
		\node [style=textbox] (0) at (-3, 0) {$\alpha^v$};
		\node [style=textbox] (1) at (-4.5, 0) {H$^u$};
		\node [style=none] (2) at (-2, 0) {};
		\node [style=none] (3) at (2, 1) {};
		\node [style=none] (4) at (2, -1) {};
		\node [style=yellownode] (5) at (0, 0) {$E^{uv}_{ab}$};
		\node [style=textbox] (6) at (4.5, 1) {$+^a$};
		\node [style=textbox] (7) at (3, 1) {H$^u$};
		\node [style=textbox] (8) at (4.5, -1) {$\alpha^v$};
		\node [style=textbox] (9) at (3, -1) {$-^b$};
	\end{pgfonlayer}
	\begin{pgfonlayer}{edgelayer}
		\draw [style=edge] (0) to (1);
		\draw [style=rededge] (2.center) to (5);
		\draw [style=rededge, bend left=15] (5) to (3.center);
		\draw [style=rededge, bend right=15] (5) to (4.center);
		\draw [style=edge] (6) to (7);
		\draw [style=edge] (8) to (9);
	\end{pgfonlayer}
\end{tikzpicture}}.
\end{center}
\end{example}

\begin{definition}[Ionic bond breaking]\label{def:iobond}
Let $u,v\in\VS$ be distinct vertex names. We define the {\em ionic bond breaking} disconnection rule $I^{uv}$ as follows:
\begin{itemize}
\item a chemical graph $A=(V,\tau,m)$ is in the domain of $I^{uv}$ if (1) $u,v\in V$, (2) $\tau(u)=+$, (3) $\tau(v)=-$, and (4) $m(u,v)=\ib$,
\item the chemical graph $I^{uv}(A)=(V,\tau,m^I)$ is defined by letting $m^I(u,v)=0$, and letting $m^I$ agree with $m$ on all other vertices.
\end{itemize}
\end{definition}
\begin{example}
The effect of an ionic bond breaking is to remove an ionic bond between two specified charges:
\begin{center}
\scalebox{.85}{\begin{tikzpicture}
	\begin{pgfonlayer}{nodelayer}
		\node [style=textbox] (1) at (-10.5, -2.75) {N};
		\node [style=textbox] (2) at (-10.5, -1.25) {H};
		\node [style=textbox] (3) at (-12, -2.5) {H};
		\node [style=textbox] (4) at (-11.75, -4) {H};
		\node [style=textbox] (5) at (-9.25, -4) {H};
		\node [style=textbox] (6) at (-9, -2.25) {$+$};
		\node [style=textbox] (7) at (-10.5, 1.75) {N};
		\node [style=textbox] (8) at (-10.5, 3.25) {H};
		\node [style=textbox] (9) at (-12, 2) {H};
		\node [style=textbox] (10) at (-11.75, 0.5) {H};
		\node [style=textbox] (11) at (-9.25, 0.5) {H};
		\node [style=textbox] (12) at (-9, 2.25) {$+^u$};
		\node [style=textbox] (13) at (-5.25, 0) {C};
		\node [style=textbox] (14) at (-3.75, 0) {O};
		\node [style=textbox] (15) at (-7.5, 2.25) {$-^v$};
		\node [style=textbox] (16) at (-6.5, 1.25) {O};
		\node [style=textbox] (17) at (-6.5, -1.25) {O};
		\node [style=textbox] (18) at (-7.5, -2.25) {$-$};
		\node [style=none] (21) at (-2.5, 0) {};
		\node [style=none] (41) at (4.75, -3) {};
		\node [style=textbox] (42) at (7, 1.75) {N};
		\node [style=textbox] (43) at (7, 3.25) {H};
		\node [style=textbox] (44) at (5.5, 2) {H};
		\node [style=textbox] (45) at (5.75, 0.5) {H};
		\node [style=textbox] (46) at (8.25, 0.5) {H};
		\node [style=textbox] (47) at (8.5, 2.25) {$+^u$};
		\node [style=textbox] (48) at (7, -2.75) {N};
		\node [style=textbox] (49) at (7, -1.25) {H};
		\node [style=textbox] (50) at (5.5, -2.5) {H};
		\node [style=textbox] (51) at (5.75, -4) {H};
		\node [style=textbox] (52) at (8.25, -4) {H};
		\node [style=textbox] (53) at (8.5, -2.25) {$+$};
		\node [style=textbox] (54) at (12.25, 0) {C};
		\node [style=textbox] (55) at (13.75, 0) {O};
		\node [style=textbox] (56) at (10, 2.25) {$-^v$};
		\node [style=textbox] (57) at (11, 1.25) {O};
		\node [style=textbox] (58) at (11, -1.25) {O};
		\node [style=textbox] (59) at (10, -2.25) {$-$};
		\node [style=turquoisenode] (60) at (0, 0) {$I^{uv}$};
		\node [style=none] (61) at (4.75, 2.25) {};
	\end{pgfonlayer}
	\begin{pgfonlayer}{edgelayer}
		\draw [style=edge] (2) to (1);
		\draw [style=edge] (1) to (3);
		\draw [style=edge] (1) to (5);
		\draw [style=edge] (4) to (1);
		\draw [style=edge] (1) to (6);
		\draw [style=edge] (8) to (7);
		\draw [style=edge] (7) to (9);
		\draw [style=edge] (7) to (11);
		\draw [style=edge] (10) to (7);
		\draw [style=edge] (7) to (12);
		\draw [style=edge] (13) to (17);
		\draw [style=edge] (16) to (13);
		\draw [style=edge] (15) to (16);
		\draw [style=edge] (17) to (18);
		\draw [style=doubedge] (14) to (13);
		\draw [style=dashed edge] (12) to (15);
		\draw [style=dashed edge] (6) to (18);
		\draw [style=edge] (43) to (42);
		\draw [style=edge] (42) to (44);
		\draw [style=edge] (42) to (46);
		\draw [style=edge] (45) to (42);
		\draw [style=edge] (42) to (47);
		\draw [style=edge] (49) to (48);
		\draw [style=edge] (48) to (50);
		\draw [style=edge] (48) to (52);
		\draw [style=edge] (51) to (48);
		\draw [style=edge] (48) to (53);
		\draw [style=edge] (54) to (58);
		\draw [style=edge] (57) to (54);
		\draw [style=edge] (56) to (57);
		\draw [style=edge] (58) to (59);
		\draw [style=doubedge] (55) to (54);
		\draw [style=dashed edge] (53) to (59);
		\draw [style=rededge, bend right=15, looseness=0.75] (61.center) to (60);
		\draw [style=rededge] (60) to (21.center);
		\draw [style=rededge, bend left, looseness=0.50] (41.center) to (60);
	\end{pgfonlayer}
\end{tikzpicture}}.
\end{center}
\end{example}

\begin{definition}[Covalent bond breaking]\label{def:covbond}
Let $u,v,a,b\in\VS$ be pairwise distinct vertex names. We define the {\em covalent bond breaking} disconnection rule $C^{uv}_{ab}$ as follows:
\begin{itemize}
\item a chemical graph $A=(V,\tau,m)$ is in the domain of $C^{uv}_{ab}$ if (1) $u,v\in V$, (2) $a,b\notin V$, (3) $\tau(u),\tau(v)\in\At\sqcup\{-\}$, and (4) $m(u,v)\in\{1,2,3,4\}$,
\item the chemical graph $C^{uv}_{ab}(A)=(V\cup\{a,b\}, \tau^C, m^C, \mathcal P^C, \mathcal T^C)$ is defined by letting $\tau^C(a)=\tau^C(b)=\alpha$ and letting $\tau^C$ agree with $\tau$ on all other vertices; further, let $m^C(u,v)=m(u,v)-1$, $m^C(u,a)=m^C(v,b)=1$ and let $m^C$ agree with $m$ on all other vertices.
\end{itemize}
\end{definition}
\begin{example}
The effect of a covalent bond breaking is to reduce the number of electron pairs in a covalent bond by one. For a single bond this results in removing the bond altogether. We give two examples of this below:
\begin{center}
\begin{tabular}{c c}
\scalebox{.75}{\begin{tikzpicture}
	\begin{pgfonlayer}{nodelayer}
		\node [style=textbox] (0) at (-6.25, 0) {C};
		\node [style=textbox] (1) at (-7.5, 1.25) {H};
		\node [style=textbox] (2) at (-4.75, 0) {C$^u$};
		\node [style=textbox] (3) at (-7.5, -1.25) {H};
		\node [style=textbox] (4) at (-3, 0) {O$^v$};
		\node [style=textbox] (8) at (4.25, 0) {C};
		\node [style=textbox] (9) at (3, 1.25) {H};
		\node [style=textbox] (11) at (3, -1.25) {H};
		\node [style=textbox] (12) at (6, 0) {C$^u$};
		\node [style=textbox] (13) at (6, 1.5) {O$^v$};
		\node [style=textbox] (14) at (7.25, -1.25) {$\alpha^a$};
		\node [style=textbox] (15) at (7.25, 2.75) {$\alpha^b$};
		\node [style=none] (16) at (-2, 0) {};
		\node [style=none] (17) at (2, 0) {};
		\node [style=rednode] (18) at (0, 0) {$C^{uv}_{ab}$};
	\end{pgfonlayer}
	\begin{pgfonlayer}{edgelayer}
		\draw [style=edge] (1) to (0);
		\draw [style=edge] (3) to (0);
		\draw [style=doubedge] (0) to (2);
		\draw [style=doubedge] (2) to (4);
		\draw [style=edge] (9) to (8);
		\draw [style=edge] (11) to (8);
		\draw [style=edge] (13) to (12);
		\draw [style=edge] (12) to (14);
		\draw [style=edge] (13) to (15);
		\draw [style=doubedge] (8) to (12);
		\draw [style=rededge] (16.center) to (18);
		\draw [style=rededge] (18) to (17.center);
	\end{pgfonlayer}
\end{tikzpicture}} \hspace{20pt} & \scalebox{.75}{\begin{tikzpicture}
	\begin{pgfonlayer}{nodelayer}
		\node [style=textbox] (0) at (-5.25, 0) {C$^v$};
		\node [style=textbox] (1) at (-6.5, 1.25) {H$^u$};
		\node [style=textbox] (2) at (-3.5, 0) {C};
		\node [style=textbox] (3) at (-6.5, -1.25) {H};
		\node [style=textbox] (4) at (-2, 0) {O};
		\node [style=none] (5) at (-1, 0) {};
		\node [style=none] (6) at (3, 1) {};
		\node [style=rednode] (7) at (1, 0) {$C^{uv}_{ab}$};
		\node [style=none] (8) at (3, -1) {};
		\node [style=textbox] (9) at (5.25, -1.75) {C$^v$};
		\node [style=textbox] (10) at (4, -0.5) {$\alpha^b$};
		\node [style=textbox] (11) at (7, -1.75) {C};
		\node [style=textbox] (12) at (4, -3) {H};
		\node [style=textbox] (13) at (8.5, -1.75) {O};
		\node [style=textbox] (14) at (4, 1) {$\alpha^a$};
		\node [style=textbox] (15) at (5.5, 1) {H$^u$};
	\end{pgfonlayer}
	\begin{pgfonlayer}{edgelayer}
		\draw [style=edge] (1) to (0);
		\draw [style=edge] (3) to (0);
		\draw [style=doubedge] (0) to (2);
		\draw [style=doubedge] (2) to (4);
		\draw [style=rededge, bend right, looseness=0.50] (7) to (8.center);
		\draw [style=rededge] (5.center) to (7);
		\draw [style=rededge, bend left, looseness=0.50] (7) to (6.center);
		\draw [style=edge] (10) to (9);
		\draw [style=edge] (12) to (9);
		\draw [style=doubedge] (9) to (11);
		\draw [style=doubedge] (11) to (13);
		\draw [style=edge] (15) to (14);
	\end{pgfonlayer}
\end{tikzpicture}}
\end{tabular}
\end{center}
\end{example}

Observe that each disconnection rule defined above is injective (as a partial function), and hence has an inverse partial function.

\section{Reactions}\label{sec:reactions}

After a disconnection rule has been applied and candidate synthetic equivalents have been found, the next step in a retrosynthetic analysis is to find an existing reaction that would transform the (hypothetical) synthetic equivalents to the target compound. In this section, we give one possible formalisation of chemical reactions using double pushout rewriting. Our approach is very similar, and inspired by, that of Andersen, Flamm, Merkle and Stadler~\cite{inferring-rule-composition}, with some important differences, such as having more strict requirements on the graphs representing molecular entities, and allowing for the placeholder variable $\alpha$.

\subsection{Morphisms}
Reactions will be defined as bottom spans of certain double pushouts~\ref{def:reaction}. Here we define the category in which these pushouts live, as well as some constructions that will help us reason about spans in this category.

\begin{definition}[Morphism of oriented pre-chemical graphs]\label{def:morphism}
A {\em morphism} of oriented pre-chemical graphs $f:A\rightarrow B$ is a function $f:V_A\rightarrow V_B$ such that its restriction to the chemical vertices $f|_{\Chem A}$ is injective, and for all $v,u,a,b\in V_A$ we have
\begin{enumerate}
\item if $v\in\Chem A$, then $\tau_B(fv)=\tau_A(v)$,
\item if $v\in\alpha(A)$, then $\tau_B(fv)\in\At\sqcup\{\alpha\}$,
\item if $v,u\in\Chem A$ and $m_A(v,u)\neq 0$, then $m_B(fv,fu)=m_A(v,u)$,
\item if $v\in\alpha(A)$ and $\cov(m_A(v,u))\neq 0$, then
$$\cov(m_B(fv,fu))=\sum_{w\in f^{-1}f(v), z\in f^{-1}f(u)} \cov(m_A(w,z)),$$
\item if $\mathcal P_A(vua)$, then $\mathcal P_B(fv,fu,fa)$,
\item if $\mathcal T_A(vuab)$, then $\mathcal T_B(fv,fu,fa,fb)$.
\end{enumerate}
\end{definition}
\begin{example}\label{ex:morphism}
The idea of a morphism is that it preserves all the atoms, bonds, charges and spatial constraints present in the domain, potentially adding new ones. We give an example below, where we use superscripts to indicate the underlying function: each vertex in the domain is mapped to the vertex in the codomain with the same superscript:
\begin{center}
\scalebox{.85}{\begin{tikzpicture}
	\begin{pgfonlayer}{nodelayer}
		\node [style=none] (17) at (-0.75, 0) {};
		\node [style=none] (18) at (1.25, 0) {};
		\node [style=textbox] (50) at (-3.75, 0) {C$^1$};
		\node [style=textbox] (51) at (-3.75, -1.5) {H$^2$};
		\node [style=textbox] (52) at (-5.25, 0) {$\alpha^3$};
		\node [style=textbox] (57) at (-3.75, 1.5) {$\alpha^4$};
		\node [style=textbox] (58) at (-2.25, 0) {$\alpha^4$};
		\node [style=textbox] (59) at (3.75, -0.5) {C$^1$};
		\node [style=textbox] (60) at (3.25, -2) {H$^2$};
		\node [style=textbox] (61) at (3.25, 1) {C$^3$};
		\node [style=textbox] (63) at (5.5, -0.5) {O$^4$};
		\node [style=textbox] (64) at (2.25, 2.5) {H};
		\node [style=textbox] (65) at (4.25, 2.5) {O};
	\end{pgfonlayer}
	\begin{pgfonlayer}{edgelayer}
		\draw [style=diredge] (17.center) to (18.center);
		\draw [style=edge] (50) to (52);
		\draw [style=edge] (50) to (57);
		\draw [style=edge] (50) to (58);
		\draw [style=edge] (59) to (61);
		\draw [style=edge] (59) to (60);
		\draw [style=doubedge] (59) to (63);
		\draw [style=edge] (64) to (61);
	\end{pgfonlayer}
\end{tikzpicture}}.
\end{center}
\end{example}

Let us denote by $\PChem$ the category of oriented pre-chemical graphs and their morphisms. This category has a partial monoidal structure given by the disjoint union: we can take the disjoint union of two morphisms provided that their domains as well as the codomains do not share vertex names. When the graphs are considered up to vertex renaming (as we shall do in the next section), this becomes an honest (strict) monoidal category.

A {\em labelled bijection} between oriented pre-chemical graphs $A$ and $C$ is a bijection $b:A\rightarrow C$ between the vertex sets which preserves the vertex labels. While a labelled bijection is not a morphism in $\PChem$, it can be understood as the span $A\hookleftarrow V_A\xrightarrow b C$, where $V_A$ is the oriented pre-chemical graph with the same vertices and vertex labels as $A$, but with all edge labels being $0$ and the triangle and tetrahedron relations being empty. The left leg of the span is the inclusion, while the right one is the bijection, now as a morphism in $\PChem$.
\begin{definition}[Intersection along a bijection]
Let $b:A\rightarrow C$ be a labelled bijection between oriented pre-chemical graphs. We define the {\em intersection along a bijection} as the oriented pre-chemical graph $A*_bC=(V_A,\tau_A,m,\mathcal P,\mathcal T)$, whose vertex set and labelling are those of $A$, and the edge labelling is defined by $m(a,a')\coloneqq m_A(a,a')$ if $m_A(a,a')=m_C(ba,ba')$ and $m(a,a')\coloneqq 0$ otherwise. The triangle and tetrahedron relations are defined by: $\mathcal P(xyz)$ if and only if $\mathcal P_A(xyz)$ and $\mathcal P_C(bx,by,bz)$, and similarly, $\mathcal T(xyzw)$ if and only if $\mathcal T_A(xyzw)$ and $\mathcal T_C(bx,by,bz,bw)$.
\end{definition}
We note that choosing the vertex set of $A*_bC$ to be $V_C$ would result in an isomorphic construction. The intersection $A*_bC$ is in fact the product of the inclusion $V_A\hookrightarrow A$ and $V_A\xrightarrow b C$ in the under category $V_A / \PChem_{inj}$, where $\PChem_{inj}$ is the wide subcategory of $\PChem$ of injective morphisms.

\subsection{Reaction Schemes}
The same reaction patterns are present in many individual reactions. A convenient way to represent this are spans whose left and right legs encode the preconditions for a reaction to occur and the effect of the reaction (outcome), respectively, while the centre denotes those parts that are unchanged.
\begin{definition}[Reaction scheme]\label{def:reaction-scheme}
A {\em reaction scheme} is a span $A\xleftarrow f K\xrightarrow g B$ in the category of pre-chemical graphs, whose boundaries $A$ and $B$ are chemical graphs with the same net charge, such that
\begin{itemize}
\item $f:V_K\rightarrow V_A$ and $g:V_K\rightarrow V_B$ are injective,
\item $f$ and $g$ are surjective on neutral vertices: if $a\in \Neu A$ and $b\in \Neu B$, then there are $k,j\in V_K$ such that $f(k)=a$ and $g(j)=b$,
\item $f$ and $g$ preserve all vertex labels: $\tau_A f=\tau_K=\tau_B g$,
\item all vertices of $K$ are neutral: $\Neu K =V_K$,
\item the span is terminal with respect to spans with the above properties.
\end{itemize}
\end{definition}

\begin{example}\label{ex:reaction-scheme}
The rule shown below appears in the equation describing glucose phosphorylation. It is a reaction scheme in the sense of Definition~\ref{def:reaction-scheme}. We denote the morphisms by vertex superscripts: the vertex in the domain is mapped to the vertex in the codomain with the same superscript.
\begin{center}
\scalebox{.85}{\begin{tikzpicture}
	\begin{pgfonlayer}{nodelayer}
		\node [style=none] (17) at (-1, 0) {};
		\node [style=none] (18) at (1, 0) {};
		\node [style=none] (19) at (-17.25, 0) {$+$};
		\node [style=none] (20) at (3.75, 0) {$+$};
		\node [style=none] (21) at (11.5, 0) {$+$};
		\node [style=textbox] (22) at (-18.75, -1.5) {$\alpha^3$};
		\node [style=textbox] (23) at (-18.75, 0) {O$^2$};
		\node [style=textbox] (24) at (-18.75, 1.5) {H$^1$};
		\node [style=textbox] (25) at (-14.25, 0) {P$^4$};
		\node [style=textbox] (26) at (-14.25, -1.5) {O$^5$};
		\node [style=textbox] (27) at (-15.75, 0) {$\alpha^6$};
		\node [style=textbox] (28) at (-14.25, 1.5) {O$^7$};
		\node [style=textbox] (29) at (-12.75, 0) {O$^8$};
		\node [style=textbox] (30) at (2.25, 0.75) {H$^1$};
		\node [style=textbox] (31) at (8.25, 0) {P$^4$};
		\node [style=textbox] (32) at (8.25, -1.5) {O$^5$};
		\node [style=textbox] (33) at (6.75, 0) {O$^2$};
		\node [style=textbox] (36) at (5.25, 0) {$\alpha^3$};
		\node [style=textbox] (37) at (13, 0.75) {$\alpha^6$};
		\node [style=textbox] (38) at (13, -0.75) {$-$};
		\node [style=textbox] (39) at (-13.25, 2.5) {$-$};
		\node [style=textbox] (40) at (-11.75, 1) {$-$};
		\node [style=textbox] (41) at (2.25, -0.75) {$+$};
		\node [style=textbox] (47) at (-7.25, -1.5) {$\alpha^3$};
		\node [style=textbox] (48) at (-7.25, 0) {O$^2$};
		\node [style=textbox] (49) at (-7.25, 1.5) {H$^1$};
		\node [style=textbox] (50) at (-4, 0) {P$^4$};
		\node [style=textbox] (51) at (-4, -1.5) {O$^5$};
		\node [style=textbox] (52) at (-5.5, 0) {$\alpha^6$};
		\node [style=none] (55) at (-8.75, 0) {};
		\node [style=none] (56) at (-10.75, 0) {};
		\node [style=textbox] (57) at (-4, 1.5) {O$^7$};
		\node [style=textbox] (58) at (-2.5, 0) {O$^8$};
		\node [style=textbox] (61) at (8.25, 1.5) {O$^7$};
		\node [style=textbox] (62) at (9.75, 0) {O$^8$};
		\node [style=textbox] (63) at (9.25, 2.5) {$-$};
		\node [style=textbox] (64) at (10.75, 1) {$-$};
	\end{pgfonlayer}
	\begin{pgfonlayer}{edgelayer}
		\draw [style=diredge] (17.center) to (18.center);
		\draw [style=edge] (22) to (23);
		\draw [style=edge] (23) to (24);
		\draw [style=edge] (27) to (25);
		\draw [style=edge] (25) to (28);
		\draw [style=edge] (25) to (29);
		\draw [style=doubedge] (25) to (26);
		\draw [style=edge] (33) to (31);
		\draw [style=doubedge] (31) to (32);
		\draw [style=edge] (36) to (33);
		\draw [style=edge] (37) to (38);
		\draw [style=edge] (28) to (39);
		\draw [style=edge] (29) to (40);
		\draw [style=edge] (30) to (41);
		\draw [style=edge] (47) to (48);
		\draw [style=doubedge] (50) to (51);
		\draw [style=diredge] (55.center) to (56.center);
		\draw [style=edge] (50) to (57);
		\draw [style=edge] (50) to (58);
		\draw [style=edge] (61) to (63);
		\draw [style=edge] (62) to (64);
		\draw [style=edge] (61) to (63);
		\draw [style=edge] (62) to (64);
		\draw [style=edge] (61) to (31);
		\draw [style=edge] (31) to (62);
	\end{pgfonlayer}
\end{tikzpicture}}.
\end{center}
\end{example}

There is a more compact way to present a reaction scheme, which suffers from the fact that the mapping involved is not a morphism in $\PChem$.
\begin{proposition}\label{prop:bij-react}
The data of a reaction scheme is in one-to-one correspondence with an ordered pair of chemical graphs $(A,B)$ with the same net charge, together with a labelled bijection $b:\Neu A\rightarrow\Neu B$.
\end{proposition}
\begin{proof}
For a reaction scheme $A\xleftarrow f K\xrightarrow g B$, we get the labelled bijection $gf^{-1}:\Neu A\rightarrow\Neu  B$.

Given the data as in the statement of the proposition, we obtain the reaction scheme
$$A\xhookleftarrow{\iota}\Neu A*_b\Neu B\xrightarrow{\kappa b}B$$
where $\iota : \Neu A\hookrightarrow A$ and $\kappa : \Neu B\hookrightarrow B$ are the inclusions.
\end{proof}

\begin{definition}[Matching]
A {\em matching} is a morphism $f:A\rightarrow C$ in $\PChem$ whose domain and codomain are both chemical graphs, such that for all $u,v\in\Chem A$ we have $m_C(fu,fv)=m_A(u,v)$.
\end{definition}
\begin{example}\label{ex:matching}
A matching is a morphism between chemical graphs (seen as pre-chemical graphs), with the further restriction that no new bonds can be added between existing vertices. We thus think of a matching as identifying the domain as a substructure of the codomain. We slightly modify the morphism in Example~\ref{ex:morphism} to obtain a matching:
\begin{center}
\scalebox{.85}{\begin{tikzpicture}
	\begin{pgfonlayer}{nodelayer}
		\node [style=none] (17) at (-0.75, 0) {};
		\node [style=none] (18) at (1.25, 0) {};
		\node [style=textbox] (50) at (-3.75, 0) {C$^1$};
		\node [style=textbox] (51) at (-3.75, -1.5) {H$^2$};
		\node [style=textbox] (52) at (-5.25, 0) {$\alpha^3$};
		\node [style=textbox] (57) at (-3.75, 1.5) {$\alpha^4$};
		\node [style=textbox] (58) at (-2.25, 0) {$\alpha^4$};
		\node [style=textbox] (59) at (3.75, -0.5) {C$^1$};
		\node [style=textbox] (60) at (3.25, -2) {H$^2$};
		\node [style=textbox] (61) at (3.25, 1) {C$^3$};
		\node [style=textbox] (63) at (5.5, -0.5) {O$^4$};
		\node [style=textbox] (64) at (2.25, 2.5) {H};
		\node [style=textbox] (65) at (4.25, 2.5) {O};
	\end{pgfonlayer}
	\begin{pgfonlayer}{edgelayer}
		\draw [style=diredge] (17.center) to (18.center);
		\draw [style=edge] (50) to (52);
		\draw [style=edge] (50) to (57);
		\draw [style=edge] (50) to (58);
		\draw [style=edge] (59) to (61);
		\draw [style=edge] (59) to (60);
		\draw [style=doubedge] (59) to (63);
		\draw [style=edge] (64) to (61);
		\draw [style=edge] (50) to (51);
		\draw [style=doubedge] (65) to (61);
	\end{pgfonlayer}
\end{tikzpicture}}.
\end{center}
\end{example}

\begin{proposition}\label{prop:matching-react}
Given a matching and a reaction scheme as below left (all morphisms are in $\PChem$), there exist unique up to an isomorphism pre-chemical graphs $D$ and $E$ such that the diagram can be completed to the one on the right, where both squares are pushouts.
\ctikzfig{reaction-matching}
Moreover, $E$ is in fact a chemical graph, and if $\alpha(C)=\eset$, then also $\alpha(E)=\eset$.
\end{proposition}
\begin{proof}[Proof sketch]
We give explicit constructions of the pre-chemical graphs $D$ and $E$, as well as the morphisms involved.

By Proposition~\ref{prop:bij-react}, we may assume that the reaction scheme is given by
$$A\xhookleftarrow{\iota}\Neu A*_b\Neu B\xrightarrow{\kappa b}B$$
for some labelled bijection $b:\Neu A\rightarrow \Neu B$. Let us denote by $(V_A,\tau_A,m_*,\mathcal P_*,\mathcal T_*)$ the intersection graph in the middle of the above diagram. Further, we assume that the vertex sets of $A$, $B$ and $C$ are disjoint. Now define the pre-chemical graph $D$ as follows:
\begin{itemize}
\item $V_D\coloneqq V_C\setminus\Crg{m(\Chem A)}$,
\item $\tau_D\coloneqq\tau_C|_{V_D}$,
\item for all $u,v\in\Neu{m(\Chem A)}$, let $m_D(u,v)\coloneqq m_*(\iota^{-1}m^{-1}u,\iota^{-1}m^{-1}v)$,
\item for all $u\in V_D\setminus m(\Chem A)$ and $v\in V_D$, let $m_D(u,v)\coloneqq m_C(u,v)$,
\item for all $x,y,z\in V_D$, if $x,y,z\in\im(m)$, then $\mathcal P_D(xyz)$ if and only if there exist $x'\in\iota^{-1}m^{-1}x$, $y'\in\iota^{-1}m^{-1}y$ and $z'\in\iota^{-1}m^{-1}z$ with $\mathcal P_*(x',y',z')$, and otherwise $\mathcal P_D(xyz)$ if and only if $\mathcal P_C(xyz)$,
\item for all $x,y,z,w\in V_D$, if $x,y,z,w\in \im(m)$, then $\mathcal T_D(xyzw)$ if and only if there exist $x'\in\iota^{-1}m^{-1}x$, $y'\in\iota^{-1}m^{-1}y$, $z'\in\iota^{-1}m^{-1}z$ and $w'\in\iota^{-1}m^{-1}w$ with $\mathcal T_*(x',y',z',w')$ and otherwise $\mathcal T_D(xyzw)$ if and only if $\mathcal T_C(xyzw)$.
\end{itemize}
In order to simplify notation in the definition of the pre-chemical graph $E$, let us introduce the following mapping
$$()^*:m(\Neu A)\cup\Crg B\rightarrow \mathcal P(V_B)$$
defined by $w^*\coloneqq bm^{-1}(w)$ if $w\in m(\Neu A)$ and $w^*\coloneqq w$ if $w\in\Crg B$. To further simplify the notation, we identify the singleton sets in $\mathcal P(V_B)$ with the corresponding elements in $V_B$ (we already use this convention in defining $()^*$ on $\Crg B$). Let us now define the pre-chemical graph $E$ as follows:
\begin{itemize}
\item $V_E\coloneqq V_D\cup\Crg B$,
\item for all $u\in V_D$, let $\tau_E(u)\coloneqq\tau_D(u)=\tau_C(u)$,
\item for all $u\in\Crg B$, let $\tau_E(u)\coloneqq\tau_B(u)$,

\item for all $x,y,z\in V_E$, if $x,y,z\in\im(m)$, then $\mathcal P_E(xyz)$ if and only if there exist $x'\in x^*$, $y'\in y^*$ and $z'\in z^*$ with $\mathcal P_B(x',y',z')$, and otherwise $\mathcal P_E(xyz)$ if and only if $\mathcal P_C(xyz)$,
\item for all $x,y,z,w\in V_E$, if $x,y,z,w\in\im(m)$, then $\mathcal T_E(xyzw)$ if and only if there exist $x'\in x^*$, $y'\in y^*$, $z'\in z^*$ and $w'\in w^*$ with $\mathcal T_B(x',y',z',w')$ and otherwise $\mathcal T_E(xyzw)$ if and only if $\mathcal T_C(xyzw)$.
\end{itemize}
Noting that the vertex set of $E$ can be written as a union of four disjoint sets
$$V_E=(V_C\setminus\im m)\cup\Neu{m(\Chem A)}\cup m(\alpha(A))\cup\Crg B,$$
we define the edge labelling function $m_E$ in the table below, where a row indicates a set to which vertex $u$ belongs, a column indicates a set to which vertex $v$ belongs, and a cell defines $m_E$ at those particular values of $u$ and $v$.
\begin{center}
\begin{tabular}{c || c | c | c | c}
$u\backslash v$ & $V_C\setminus\im m$ & $\Neu{m(\Chem A)}$ & $m(\alpha(A))$ & $\Crg B$ \\ \hhline{=#=|=|=|=}
$V_C\setminus\im m$ & $m_C(u,v)$ & $m_C(u,v)=0$ & $m_C(u,v)$ & $0$ \\ \hline
$\Neu{m(\Chem A)}$ & $m_C(u,v)=0$ & $m_B(u^*,v^*)$ & $\sum_{w\in v^*} m_B(u^*,w)$ & $m_B(u^*,v)$ \\ \hline
$m(\alpha(A))$ & $m_C(u,v)$ & $\sum_{w\in u^*} m_B(v^*,w)$ & $m_C(u,v)$ & $\sum_{w\in u^*} m_B(v,w)$ \\ \hline
$\Crg B$ & $0$ & $m_B(v^*,u)$ & $\sum_{w\in v^*} m_B(u,w)$ & $m_B(u,v)$
\end{tabular}
\end{center}
One then checks that the squares below are pushouts:
\ctikzfig{matching-pbc}
where $\iota'$ and $\kappa'$ are the appropriate inclusions, and $m''$ is defined as $m\iota b^{-1}$ on $\Neu B$ and as identity on $\Crg B$.
\end{proof}

\begin{definition}\label{def:reaction}
A {\em reaction} is a double pushout diagram
\ctikzfig{reaction}
in $\PChem$ such that the top span $A\leftarrow K\rightarrow B$ is a reaction scheme, $m:A\rightarrow C$ is a matching, and $C$ and $E$ are oriented molecular graphs (cf.~Proposition~\ref{prop:matching-react}). We say that $C$ is the {\em reactant graph} and $E$ the {\em product graph} of the reaction.
\end{definition}
Analogously to reaction schemes, we can present a reaction in a more concrete (yet equivalent) way (Proposition~\ref{prop:reaction-concrete}), which involves mappings that are not morphisms in $\PChem$. To this end, we need the notion of a {\em chemical subgraph}.
\begin{definition}[Chemical subgraph]
Given a chemical graph $A$, a {\em chemical subgraph} is a subset $U\sse V_A$ such that (1) if for $u\in U$ one has $\tau_A(u)\in\{+,\alpha\}$ and $v\in V_A$ is the unique vertex with $m_A(u,v)=1$, then $v\in U$, and (2) if for $u\in U$ and $v\in\Crg A$ one has $m_A(u,v)=1$, then $v\in U$.
\end{definition}
Note that a chemical subgraph is not itself a chemical graph, albeit it can be converted to one by adding an appropriate number of $\alpha$-vertices, as we shall do in the proof of Proposition~\ref{prop:reaction-concrete}. When reasoning about chemical subgraphs, the following concept will be useful.
\begin{definition}[Covalent neighbours]
Given a pre-chemical graph $A$ and a vertex $v\in V_A$, we define the set of {\em covalent neighbours} of $v$ as follows:
$$CN(v)\coloneqq\{u\in V_A : \cov(m_A(v,u))\neq 0\}.$$
\end{definition}
\begin{lemma}\label{lma:matching-equal-images}
If $m:A\rightarrow C$ is a matching and $v\in\Chem A$, then $m(CN(v))=CN(mv)$.
\end{lemma}
\begin{proof}
The inclusion $m(CN(v))\sse CN(mv)$ is clear. On the other hand, since $\tau_A(v)=\tau_C(mv)$, we have
$$\sum_{u\in CN(v)} m_A(v,u) = \sum_{u\in CN(mv)} m_C(mv,u).$$
Using the fact that $m$ is a matching, one can further show that
$$\sum_{u\in CN(v)} m_A(v,u) = \sum_{u\in m(CN(v))} m_C(mv,u),$$
so that we obtain
$$\sum_{u\in m(CN(v))} m_C(mv,u) = \sum_{u\in CN(mv)} m_C(mv,u).$$
Since $m(CN(v))\sse CN(mv)$, it follows that the sets are, in fact, equal.
\end{proof}
\begin{lemma}\label{lma:image-matching-subgraph}
If $m:A\rightarrow C$ is a matching, then $m(\Chem A)$ is a chemical subgraph of $C$.
\end{lemma}
\begin{proof}
Since $m(\Chem A)$ contains no $\alpha$-vertices, the condition of being a chemical graph becomes: (1) if $u\in m(\Chem A)$ and $v\in V_C$ such that $\tau_C(u)=+$ and $m_C(u,v)=1$, then $v\in m(\Chem A)$, and (2) for any $u\in m(\Chem A)$ and $v\in\Crg C$ with $m_C(u,v)=1$, we have that also $v\in m(\Chem A)$.

For (1), let $u\in m(\Chem A)$ and $v\in V_C$ with $\tau_C(u)=+$ and $m_C(u,v)=1$. Let $u'\in V_A$ be the unique vertex with $mu'=u$. Since $\tau_A(u')=+$, there is a unique $v'\in\Chem A$ with $m_A(u',v')=1$. It follows that $mv'=v$, so that $v\in m(\Chem A)$.

For (2), suppose that $u\in m(\Chem A)$ and $v\in\Crg C$ with $m_C(u,v)=1$. Let $u'\in V_A$ be the unique vertex with $mu'=v$. By Lemma~\ref{lma:matching-equal-images}, $m(CN(u'))=CN(u)$, so that $v\in m(CN(u'))$. Since an $\alpha$-vertex cannot be mapped to a $v$ (as it is charged), it follows that $v\in m(\Chem A)$.
\end{proof}
\begin{proposition}\label{prop:reaction-concrete}
Let $C$ and $E$ be molecular graphs. The data of a reaction $C\rightarrow E$ can be equivalently presented as a tuple $(U_C,U_E,b,i)$ where $U_C\sse V_C$ and $U_E\sse V_E$ are chemical subgraphs with equal net charge, $b:\Neu{U_C}\rightarrow \Neu{U_E}$ is a labelled bijection, and $i:V_C\setminus U_C\rightarrow V_E\setminus U_E$ is an isomorphism of pre-chemical graphs.
\end{proposition}
\begin{proof}
Given a reaction
\ctikzfig{reaction}
we obtain the tuple $(m(\Chem A),m''(\Chem B),g'(f')^{-1},g'(f')^{-1})$, where the first occurrence of $g'(f')^{-1}$ is understood to be restricted to $\Neu{m(\Chem A)}\rightarrow\Neu{m''(\Chem B)}$, while the second has the type $V_C\setminus m(\Chem A)\rightarrow V_E\setminus m''(\Chem B)$.

Conversely, given a tuple as in the statement of the proposition, we obtain the following reaction
\ctikzfig{tuple-to-reaction}
whose details we now define. First, let us define the following set of formal symbols:
$$N(U_C)\coloneqq\{v^u_j : u\in U_C, v\in CN(u)\cap (V_C\setminus U_C)\text{ and } j=1,\dots,m_C(u,v)\}.$$
We define the chemical graph $U_C^{\alpha}\coloneqq(U_C\cup N(U_C),\tau^{\alpha},m^{\alpha})$\footnote{We assume a unique choice of a vertex name for every element in $N(U_C)$, disjoint from the vertex names of $V_C$. For legibility, we omit this technical detail.} by the following labelling functions: for $u,w\in U_C$ and $v^a_j,z^b_i\in N(U_C)$, let $\tau^{\alpha}(u)\coloneqq\tau_C(u)$, $\tau^{\alpha}(v^a_j)\coloneqq\alpha$, $m^{\alpha}(u,w)\coloneqq m_C(u,w)$, $m^{\alpha}(v^a_j,z^b_i)\coloneqq 0$, $m^{\alpha}(u,v^a_j)\coloneqq 1$ if $u=a$ and $m^{\alpha}(u,v^a_j)\coloneqq 0$ if $u\neq a$. The chemical graph $U_E^{\alpha}$ is defined similarly.

The matching $m:U_C^{\alpha}\rightarrow C$ is defined as identity on $U_C$, and for $v^u_j\in N(U_C)$ as $v^u_j\mapsto v$. The matching $m''$ is defined similarly.

The labelled bijection $b^{\alpha}:\Neu{U_C^{\alpha}}\rightarrow\Neu{U_E^{\alpha}}$ is defined by: for $u\in\Neu{U_C}$, let $u\mapsto b(u)$, and for each $v\in V_C\setminus U_C$ the set $\{v^u_j : v^u_j\in N(U_C)\}$ is mapped to the set $\{i(v)^u_j : i(v)^u_j\in N(U_E)\}$, which have the same size since $i$ is an isomorphism.

The maps $\iota$, $\kappa$ and $\iota'$ are the appropriate inclusions, while $b+i$ is the coproduct function of $b:\Neu{U_C}\rightarrow V_E$ and $i:V_C\setminus U_C\rightarrow V_E$.
\end{proof}
\begin{example}
The following reaction (glucose phosphorylation) is an instance of the reaction scheme in Example~\ref{ex:reaction-scheme}; we have labelled the vertices in the images of matchings on both sides:
\begin{center}
\scalebox{.65}{\input{phosphorylation.tikz}}.
\end{center}
\end{example}
\begin{definition}[Category of reactions]
We denote by $\React$ the {\em category of reactions}, whose
\begin{itemize}
\item objects are molecular graphs,
\item morphisms $A\rightarrow B$ are tuples $(U_A,U_B,b,i)$ as in Proposition~\ref{prop:reaction-concrete},
\item the composition of $(U_A,U_B,b,i):A\rightarrow B$ and $(W_B,W_C,c,j):B\rightarrow C$ is given by
$$(Z_A,Z_C,(c+j)(b+i),ji):A\rightarrow C,$$
where $Z_A\coloneqq U_A\cup i^{-1}(W_B\cap (V_B\setminus U_B))$ and $Z_C\coloneqq W_C\cup j(U_B\cap (V_B\setminus W_B))$,
\item for a molecular graph $A$, the identity is given by $(\eset,\eset,!,\id_A)$, where $!$ is the unique endomorphism on the empty set.
\end{itemize}
\end{definition}
\begin{remark}
We can compose the reactions as double pushout diagrams (Definition~\ref{def:reaction}) if we view them as 1-cells in the following 2-category:
\begin{itemize}
\item the 0-cells are molecular graphs,
\item a 1-cell $A\rightarrow C$ is a reaction with the reactant graph $A$ and the product graph $C$,
\item a 2-cell between two parallel reactions consists of four morphisms: a morphism between the apexes of the bottom spans, and a span morphism between the top spans, such that all resulting diagrams commute,
\item the identity reaction on $A$ is given by the empty reaction scheme,
\item the horizontal composition of two reactions is given as follows:
\begin{enumerate}
\item form the pullbacks marked in the diagram below,
\item form the pushout of the upper pullback,
\item one can show that the pushout embeds into the lower pullback (dashed line),
\item complete the double pushout diagram by forming pushout complements.
\end{enumerate}
\end{itemize}
\begin{center}
\scalebox{.8}{\begin{tikzpicture}
	\begin{pgfonlayer}{nodelayer}
		\node [style=object] (22) at (0, -4.5) {$\bullet$};
		\node [style=object] (28) at (-2.5, -1) {$\bullet$};
		\node [style=object] (7) at (-9.5, 1) {$\bullet$};
		\node [style=object] (8) at (-5.5, 1) {$\bullet$};
		\node [style=object] (9) at (-1.5, 1) {$\bullet$};
		\node [style=object] (10) at (-9.5, -2) {$\bullet$};
		\node [style=object] (14) at (-5.5, -2) {$\bullet$};
		\node [style=object] (16) at (1.5, 1) {$\bullet$};
		\node [style=object] (17) at (5.5, 1) {$\bullet$};
		\node [style=object] (18) at (9.5, 1) {$\bullet$};
		\node [style=object] (19) at (0, -2) {$\bullet$};
		\node [style=object] (20) at (5.5, -2) {$\bullet$};
		\node [style=object] (21) at (9.5, -2) {$\bullet$};
		\node [style=none] (23) at (0, -4) {\hspace{-2pt}\scalebox{1.2}{\rotatebox{-45}{$\ulcorner$}}};
		\node [style=object] (26) at (0, 3.5) {$\bullet$};
		\node [style=none] (27) at (0, 3) {\hspace{3pt}\scalebox{1.2}{\rotatebox{135}{$\ulcorner$}}};
		\node [style=none] (29) at (-2.5, -0.5) {\hspace{-2pt}\scalebox{1.2}{\rotatebox{-45}{$\ulcorner$}}};
	\end{pgfonlayer}
	\begin{pgfonlayer}{edgelayer}
		\draw [style=dash morphism] (28) to (22);
		\draw [style=morphism] (8) to (7);
		\draw [style=morphism] (8) to (9);
		\draw [style=morphism] (7) to (10);
		\draw [style=morphism] (8) to (14);
		\draw [style=morphism] (14) to (10);
		\draw [style=morphism] (17) to (16);
		\draw [style=morphism] (17) to (18);
		\draw [style=morphism] (16) to (19);
		\draw [style=morphism] (17) to (20);
		\draw [style=morphism] (18) to (21);
		\draw [style=morphism] (20) to (19);
		\draw [style=morphism] (20) to (21);
		\draw [style=morphism] (9) to (19);
		\draw [style=morphism] (14) to (19);
		\draw [style=morphism] (22) to (14);
		\draw [style=morphism] (22) to (20);
		\draw [style=morphism] (26) to (8);
		\draw [style=morphism] (26) to (17);
		\draw [style=morphism] (8) to (28);
		\draw [style=morphism] (17) to (28);
	\end{pgfonlayer}
\end{tikzpicture}}
\end{center}
In this light, we can view Proposition~\ref{prop:reaction-concrete} as showing that the 1-categorical part of the 2-category sketched above is equivalent to $\React$.
\end{remark}

\section{Retrosynthesis in Layered Props}\label{sec:retrosynthesis}

The main object of interest of this paper is the layered prop whose layers all share the same set of objects: namely, the chemical graphs up to a labelled graph isomorphism. The morphisms of a layer are either matchings, disconnection rules or reactions, parameterised by environmental molecules (these can act as solvents, reagents or catalysts). These layers are the main building blocks of our formulation of retrosynthesis.

Given a finite set $M$ of molecular entities, let us enumerate the molecular entities in $M$ as $M_1,\dots,M_k$. Given a list natural numbers $n=(n_1,\dots,n_k)$, we denote the resulting molecular graph $n_1M_1 + \dots + n_kM_k$ by $(V_n,\tau_n,m_n)$. We define three classes of symmetric monoidal categories parameterised by $M$ as follows. The objects for all categories are the (equivalence classes of) chemical graphs, and the morphisms $A\rightarrow B$ are given below:
\paragraph{$\MMatch$:} a morphism $(m,b) : A\rightarrow B$ is given by a matching $m:A\rightarrow B$ together with a labelled injection $b:b_1M_1 + \dots + b_kM_k\rightarrow B$ such that $\im(m)\cup\im(b) = B$, and $\im(m)\cap\im(b)=m(\alpha(A))\cap\Chem B$; the composite $A\xrightarrow{m,b}B\xrightarrow{n,c}C$ is given by $nm:A\rightarrow C$ and $nb+c:(b_1+c_1)M_1 + \dots + (b_k+c_k)M_k\rightarrow C$.
\paragraph{$\MReact$:} a generating morphism is a reaction $n_1M_1 + \dots + n_kM_k + A\xrightarrow r B$; given another reaction $m_1M_1 + \dots + m_kM_k + B\xrightarrow s C$, the composite $A\rightarrow C$ is given by
$$s\circ (r + \id_{m_1M_1 + \dots + m_kM_k}) : (n_1+m_1)M_1 + \dots + (n_k+m_k)M_k + A\rightarrow C.$$
\paragraph{$\MDisc$:} for every disconnection rule $d^{uv}_{ab}$ such that $d^{uv}_{ab}(n_1M_1 + \dots + n_kM_k + A)=B$, there are generating morphisms $d^{uv}_{ab}:A\rightarrow B$ and $\bar d^{uv}_{ab}:B\rightarrow A$, subject to the following equations:
\begin{itemize}
\item $\bar d^{uv}_{ab}d^{uv}_{ab}=\id_A$ and $d^{uv}_{ab}\bar d^{uv}_{ab}=\id_B$,
\item $d^{uv}_{ab}h^{wz}_{xy}=h^{wz}_{xy}d^{uv}_{ab}$ whenever both sides are defined,
\item $d_{u,v}+\id_C=d_{u,v}$ for every chemical graph $C$.
\end{itemize}

The idea is that the set $M$ models the reaction environment: the parametric definitions above capture the intuition that there is an unbounded supply of these molecules in the environment. In order to interpret sequences of disconnection rules as reactions, we need to restrict to those sequences whose domain and codomain are both molecular graphs: we thus write $\MDiscMol$ for the full subcategory of $M\mhyphen\Disc$ on molecular graphs. If $M=\eset$, we may omit the prefix.
\begin{example}\label{ex:mmatch}
A morphism in $\MMatch$ is a matching such that the environment contains enough ``building material'' to cover the complement of the image of the matching. We give an example below, taking $M=\{\texttt{CH\textsubscript{2}OHO\textsuperscript{+}}\}$, where the horizontal map is the matching, and the vertical map is the labelled injection:
\begin{center}
\scalebox{0.7}{\begin{tikzpicture}
	\begin{pgfonlayer}{nodelayer}
		\node [style=none] (17) at (-3.75, -1.5) {};
		\node [style=none] (18) at (-1.75, -1.5) {};
		\node [style=textbox] (50) at (-6.75, -1.5) {C$^1$};
		\node [style=textbox] (51) at (-6.75, -3) {$\alpha^2$};
		\node [style=textbox] (52) at (-8.25, -1.5) {$\alpha^3$};
		\node [style=textbox] (57) at (-6.75, 0) {$\alpha^4$};
		\node [style=textbox] (58) at (-5.25, -1.5) {$\alpha^4$};
		\node [style=textbox] (59) at (0.75, -2) {C$^1$};
		\node [style=textbox] (60) at (0.25, -3.5) {H$^2$};
		\node [style=textbox] (61) at (0.25, -0.5) {C$^3$};
		\node [style=textbox] (63) at (2.5, -2) {O$^4$};
		\node [style=textbox] (64) at (-0.75, 1) {H$^5$};
		\node [style=textbox] (65) at (1.25, 1) {O$^6$};
		\node [style=textbox] (66) at (2.75, 6.75) {C$^3$};
		\node [style=textbox] (67) at (1.5, 8.25) {H$^5$};
		\node [style=textbox] (68) at (4, 8.25) {O$^6$};
		\node [style=textbox] (72) at (5.25, -0.5) {$+^7$};
		\node [style=textbox] (73) at (5.25, -2) {H$^8$};
		\node [style=none] (74) at (2.75, 4) {};
		\node [style=none] (75) at (2.75, 2) {};
		\node [style=textbox] (77) at (5.75, 5.25) {H$^8$};
		\node [style=textbox] (78) at (4, 5.25) {O$^4$};
		\node [style=textbox] (79) at (5.75, 8.25) {$+^7$};
		\node [style=textbox] (80) at (1.5, 5.25) {H$^2$};
	\end{pgfonlayer}
	\begin{pgfonlayer}{edgelayer}
		\draw [style=diredge] (17.center) to (18.center);
		\draw [style=edge] (50) to (52);
		\draw [style=edge] (50) to (57);
		\draw [style=edge] (50) to (58);
		\draw [style=edge] (59) to (61);
		\draw [style=edge] (59) to (60);
		\draw [style=doubedge] (59) to (63);
		\draw [style=edge] (64) to (61);
		\draw [style=edge] (50) to (51);
		\draw [style=doubedge] (65) to (61);
		\draw [style=edge] (67) to (66);
		\draw [style=edge] (72) to (73);
		\draw [style=diredge] (74.center) to (75.center);
		\draw [style=edge] (66) to (78);
		\draw [style=edge] (78) to (77);
		\draw [style=edge] (66) to (68);
		\draw [style=edge] (68) to (79);
		\draw [style=edge] (80) to (66);
	\end{pgfonlayer}
\end{tikzpicture}}.
\end{center}
\end{example}
\begin{example}\label{ex:mdisc}
A morphism in $\MDisc$ is a chain of (dis)connection rules. We give an example below with $M=\eset$:
\begin{center}
\scalebox{0.7}{\begin{tikzpicture}
	\begin{pgfonlayer}{nodelayer}
		\node [style=textbox] (0) at (-16.75, 7) {C$^z$};
		\node [style=textbox] (1) at (-18, 8.25) {H$^w$};
		\node [style=textbox] (2) at (-15.25, 7) {C$^u$};
		\node [style=textbox] (3) at (-18, 5.75) {H};
		\node [style=textbox] (4) at (-13.75, 7) {O$^v$};
		\node [style=textbox] (8) at (-6.5, 7) {C$^z$};
		\node [style=textbox] (9) at (-7.75, 8.25) {H$^w$};
		\node [style=textbox] (10) at (-7.75, 5.75) {H};
		\node [style=textbox] (11) at (-5, 7) {C$^u$};
		\node [style=textbox] (12) at (-5, 8.5) {O$^v$};
		\node [style=textbox] (13) at (-3.75, 5.75) {$\alpha^a$};
		\node [style=textbox] (14) at (-3.75, 9.75) {$\alpha^b$};
		\node [style=textbox] (24) at (3.5, 8.75) {$\alpha^c$};
		\node [style=textbox] (25) at (5, 8.75) {H$^w$};
		\node [style=textbox] (26) at (3.5, 4.25) {C$^z$};
		\node [style=textbox] (27) at (2.25, 5.5) {$\alpha^d$};
		\node [style=textbox] (28) at (5, 4.25) {C$^u$};
		\node [style=textbox] (29) at (2.25, 3) {H};
		\node [style=textbox] (30) at (5, 5.75) {O$^v$};
		\node [style=textbox] (31) at (6.25, 3) {$\alpha^a$};
		\node [style=textbox] (32) at (6.25, 7) {$\alpha^b$};
		\node [style=none] (33) at (-12.5, 7) {};
		\node [style=none] (34) at (-8.5, 7) {};
		\node [style=rednode] (35) at (-10.5, 7) {$C^{uv}_{ab}$};
		\node [style=none] (36) at (-2.75, 7) {};
		\node [style=none] (37) at (2, 8.75) {};
		\node [style=rednode] (38) at (-0.75, 7) {$C^{wz}_{cd}$};
		\node [style=none] (39) at (1.25, 4.5) {};
		\node [style=none] (40) at (7.75, 4.75) {};
		\node [style=none] (41) at (12.75, 7.25) {};
		\node [style=none] (42) at (12.25, 3.5) {};
		\node [style=yellownode] (43) at (9.75, 4.75) {$E^{zd}_{ef}$};
		\node [style=none] (44) at (7, 8.75) {};
		\node [style=none] (45) at (12.75, 8.75) {};
		\node [style=textbox] (46) at (14.25, 8.75) {$\alpha^c$};
		\node [style=textbox] (47) at (15.75, 8.75) {H$^w$};
		\node [style=textbox] (48) at (14.5, 3.5) {C$^z$};
		\node [style=textbox] (49) at (13.25, 4.75) {$+^e$};
		\node [style=textbox] (50) at (16, 3.5) {C$^u$};
		\node [style=textbox] (51) at (13.25, 2.25) {H};
		\node [style=textbox] (52) at (16, 5) {O$^v$};
		\node [style=textbox] (53) at (17.25, 2.25) {$\alpha^a$};
		\node [style=textbox] (54) at (17.25, 6.25) {$\alpha^b$};
		\node [style=textbox] (55) at (16, 7.25) {$-^f$};
		\node [style=textbox] (56) at (14.25, 7.25) {$\alpha^d$};
		\node [style=none] (57) at (-8.5, -4) {};
		\node [style=none] (58) at (-12.5, -2.75) {};
		\node [style=rednode] (59) at (-10.5, -4) {$\bar C^{fv}_{db}$};
		\node [style=none] (60) at (-12.5, -5.75) {};
		\node [style=none] (61) at (-12.5, -1.25) {};
		\node [style=none] (62) at (-7, -1.25) {};
		\node [style=textbox] (63) at (-5.75, -1.25) {$\alpha^c$};
		\node [style=textbox] (64) at (-4.25, -1.25) {H$^w$};
		\node [style=textbox] (65) at (-6.25, -5.25) {C$^z$};
		\node [style=textbox] (66) at (-7.5, -4) {$+^e$};
		\node [style=textbox] (67) at (-4.75, -5.25) {C$^u$};
		\node [style=textbox] (68) at (-7.5, -6.5) {H};
		\node [style=textbox] (69) at (-4.75, -3.75) {O$^v$};
		\node [style=textbox] (70) at (-3.5, -6.5) {$\alpha^a$};
		\node [style=textbox] (71) at (-3.5, -2.5) {$-^f$};
		\node [style=none] (72) at (2.25, -2.5) {};
		\node [style=none] (73) at (-3, -1.25) {};
		\node [style=rednode] (74) at (0.25, -2.5) {$\bar C^{wu}_{ca}$};
		\node [style=none] (75) at (-2.5, -4.75) {};
		\node [style=textbox] (76) at (4.5, -3.25) {C$^z$};
		\node [style=textbox] (77) at (3.25, -2) {$+^e$};
		\node [style=textbox] (78) at (6, -3.25) {C$^u$};
		\node [style=textbox] (79) at (3.25, -4.5) {H};
		\node [style=textbox] (80) at (6, -1.75) {O$^v$};
		\node [style=textbox] (81) at (7.25, -4.5) {H$^w$};
		\node [style=textbox] (82) at (7.25, -0.5) {$-^f$};
	\end{pgfonlayer}
	\begin{pgfonlayer}{edgelayer}
		\draw [style=edge] (1) to (0);
		\draw [style=edge] (3) to (0);
		\draw [style=doubedge] (0) to (2);
		\draw [style=doubedge] (2) to (4);
		\draw [style=edge] (9) to (8);
		\draw [style=edge] (10) to (8);
		\draw [style=edge] (12) to (11);
		\draw [style=edge] (11) to (13);
		\draw [style=edge] (12) to (14);
		\draw [style=doubedge] (8) to (11);
		\draw [style=edge] (25) to (24);
		\draw [style=edge] (27) to (26);
		\draw [style=edge] (29) to (26);
		\draw [style=edge] (30) to (28);
		\draw [style=doubedge] (26) to (28);
		\draw [style=edge] (28) to (31);
		\draw [style=edge] (30) to (32);
		\draw [style=rededge] (33.center) to (35);
		\draw [style=rededge] (35) to (34.center);
		\draw [style=rededge, bend right, looseness=0.50] (38) to (39.center);
		\draw [style=rededge] (36.center) to (38);
		\draw [style=rededge, bend left, looseness=0.50] (38) to (37.center);
		\draw [style=rededge] (40.center) to (43);
		\draw [style=rededge, bend left=15] (43) to (41.center);
		\draw [style=rededge, bend right=15] (43) to (42.center);
		\draw [style=rededge] (44.center) to (45.center);
		\draw [style=edge] (47) to (46);
		\draw [style=edge] (49) to (48);
		\draw [style=edge] (51) to (48);
		\draw [style=edge] (52) to (50);
		\draw [style=doubedge] (48) to (50);
		\draw [style=edge] (50) to (53);
		\draw [style=edge] (52) to (54);
		\draw [style=edge] (56) to (55);
		\draw [style=rededge, bend left, looseness=0.50] (59) to (60.center);
		\draw [style=rededge] (57.center) to (59);
		\draw [style=rededge, bend right=15, looseness=0.75] (59) to (58.center);
		\draw [style=rededge] (61.center) to (62.center);
		\draw [style=edge] (64) to (63);
		\draw [style=edge] (66) to (65);
		\draw [style=edge] (68) to (65);
		\draw [style=edge] (69) to (67);
		\draw [style=doubedge] (65) to (67);
		\draw [style=edge] (67) to (70);
		\draw [style=edge] (69) to (71);
		\draw [style=rededge, bend left, looseness=0.50] (74) to (75.center);
		\draw [style=rededge] (72.center) to (74);
		\draw [style=rededge, bend right=15, looseness=0.75] (74) to (73.center);
		\draw [style=edge] (77) to (76);
		\draw [style=edge] (79) to (76);
		\draw [style=edge] (80) to (78);
		\draw [style=doubedge] (76) to (78);
		\draw [style=edge] (78) to (81);
		\draw [style=edge] (80) to (82);
	\end{pgfonlayer}
\end{tikzpicture}}.
\end{center}
Since the endpoints of the morphism are molecular graphs, the morphism is, in fact, in $\DiscMol$.
\end{example}

Next, we define the following identity-on-object functors between the above parameterised categories:
\begin{equation}\label{eq:m-functors}
\scalebox{1}{\begin{tikzpicture}
	\begin{pgfonlayer}{nodelayer}
		\node [style=objectbox] (0) at (-4.25, 2.25) {$M\mhyphen\Match$};
		\node [style=objectbox] (1) at (4.75, -1.75) {$M\mhyphen\React$};
		\node [style=objectbox] (2) at (-1.25, -1.75) {$M\mhyphen\Disc$};
		\node [style=objectbox] (3) at (1.75, 2.25) {$M\mhyphen\DiscMol$};
		\node [style=none] (4) at (-3.5, 0.25) {$D$};
		\node [style=none] (5) at (3.75, 0.25) {$R$};
	\end{pgfonlayer}
	\begin{pgfonlayer}{edgelayer}
		\draw [style=morphism] (0) to (2);
		\draw [style=hookarrow] (3) to (2);
		\draw [style=morphism] (3) to (1);
	\end{pgfonlayer}
\end{tikzpicture}}.
\end{equation}
Given a morphism $(m,b) : A\rightarrow B$ in $\MMatch$, let us denote by $M_b=(V_b,\tau_b,m_b)$ the domain molecular graph of $b$, and by $(V_{Ab},\tau_{Ab},m_{Ab})$ the disjoint union of $A$ and $M_b$. Let $()^*:V_{Ab}\rightarrow V_B$ be the induced map. We define $D(m,b)\in\MDisc$ by
\begin{align*}
&\prod_{u,v\in\tau^{-1}_{b}\{+,-\}}\left(I^{uv}\right)^{\ion(m_b(u,v))}; \prod_{u,v\in\tau^{-1}_{b}(\At\coprod \{-\})}\prod_{i=1}^{\cov(m_b(u,v))}C^{uv}_{a_ib_i};  \\
&\prod_{u\in\tau^{-1}_{b}(+), v,w\in\tau^{-1}_{b}(\At)}\left(E^{vc}_{ab}\right)^{\cov(m_B(bu,bv))};\left(\bar E^{wc}_{ub}\right)^{\cov(m_b(u,w))}; \\
&\prod_{u,v\in\tau^{-1}_{Ab}(\At\coprod \{-\})}\prod_{i=1}^{\cov(m_B(u^*,v^*))}\bar C^{uv}_{a_ib_i}; \prod_{u,v\in\tau^{-1}_{Ab}\{+,-\}}\left(\bar I^{uv}\right)^{\ion(m_B(u^*,v^*))},
\end{align*}
where on the second line $c$ ranges over all $\alpha$-vertices.

Given a disconnection rule $d^{uv}_{ab}:A\rightarrow B$, let $CG\{u,v\}$ be the smallest chemical subgraph of $A$ containing $u$ and $v$. Then $CG\{u,v\}\cup\{a,b\}$ is a chemical subgraph of $B$. Since $\{a,b\}$ are either $\alpha$-vertices (covalent bond breaking), charges with opposite signs (electron detachment) or empty (ionic bond breaking), the sets $CG\{u,v\}$ and $CG\{u,v\}\cup\{a,b\}$ have the same net charge, and there is a labelled bijection between $CG\{u,v\}$ and the {\em chemical} neutral vertices of $CG\{u,v\}\cup\{a,b\}$. Since the disconnection rule does not change any other vertices or edges, the pre-chemical graphs $A\setminus CG\{u,v\}$ and $B\setminus\left(CG\{u,v\}\cup\{a,b\}\right)$ are equal. Thus, we define the functor $R:\MDiscMol\rightarrow\MReact$ by sending a sequence of disconnection rules to the composite of the above labelled bijections and isomorphisms.

Additionally, for every pair of finite sets of molecular entities such that $M\sse N$, there is an inclusion functor for each of the four classes of categories.

\begin{example}
Applying $D$ to the morphism in Example~\ref{ex:mmatch} yields the following sequence of (dis)connection rules in $\{\texttt{CH\textsubscript{2}OHO\textsuperscript{+}}\}\mhyphen\Disc$, where we combine the commuting rules of the same type into a single node:
\begin{center}
\scalebox{0.7}{\begin{tikzpicture}
	\begin{pgfonlayer}{nodelayer}
		\node [style=textbox] (50) at (-12, 12.5) {C$^1$};
		\node [style=textbox] (51) at (-12, 11) {$\alpha^2$};
		\node [style=textbox] (52) at (-13.5, 12.5) {$\alpha^3$};
		\node [style=textbox] (57) at (-12, 14) {$\alpha^4$};
		\node [style=textbox] (58) at (-10.5, 12.5) {$\alpha^5$};
		\node [style=textbox] (59) at (24.25, 6.75) {C$^1$};
		\node [style=textbox] (60) at (23.75, 5.25) {H$^7$};
		\node [style=textbox] (61) at (23.75, 8.25) {C$^6$};
		\node [style=textbox] (63) at (26, 6.75) {O$^9$};
		\node [style=textbox] (64) at (22.75, 9.75) {H$^8$};
		\node [style=textbox] (65) at (24.75, 9.75) {O$^{10}$};
		\node [style=textbox] (66) at (-13.25, 6) {C$^6$};
		\node [style=textbox] (67) at (-14.5, 7.5) {H$^8$};
		\node [style=textbox] (68) at (-12, 7.5) {O$^{10}$};
		\node [style=textbox] (72) at (23.75, 2.75) {$+^{24}$};
		\node [style=textbox] (73) at (23.75, 1.25) {H$^{11}$};
		\node [style=textbox] (77) at (-10.25, 4.5) {H$^{11}$};
		\node [style=textbox] (78) at (-12, 4.5) {O$^9$};
		\node [style=textbox] (79) at (-10.25, 7.5) {$+^{12}$};
		\node [style=textbox] (80) at (-14.5, 4.5) {H$^7$};
		\node [style=rednode] (81) at (-7.5, 6) {$C$};
		\node [style=none] (82) at (-9.5, 6) {};
		\node [style=textbox] (83) at (-2.5, 10.25) {C$^6$};
		\node [style=textbox] (84) at (-3.75, 11.75) {$\alpha^{14}$};
		\node [style=textbox] (85) at (-1.25, 11.75) {$\alpha^{15}$};
		\node [style=textbox] (86) at (-1.25, 8.75) {$\alpha^{16}$};
		\node [style=textbox] (87) at (-3.75, 8.75) {$\alpha^{13}$};
		\node [style=textbox] (88) at (-1.5, 7.25) {$\alpha^{18}$};
		\node [style=textbox] (89) at (-3.75, 7.25) {H$^7$};
		\node [style=textbox] (90) at (-1.5, 6) {$\alpha^{19}$};
		\node [style=textbox] (91) at (-3.75, 6) {H$^8$};
		\node [style=textbox] (92) at (-3.75, 3.5) {$\alpha^{22}$};
		\node [style=textbox] (93) at (-1.5, 3.5) {O$^{10}$};
		\node [style=textbox] (94) at (0.5, 3.5) {$+^{12}$};
		\node [style=textbox] (95) at (-3.75, 4.75) {$\alpha^{20}$};
		\node [style=textbox] (96) at (0.5, 4.75) {$\alpha^{21}$};
		\node [style=textbox] (97) at (-1.5, 4.75) {O$^9$};
		\node [style=textbox] (98) at (-1.5, 0.25) {H$^{11}$};
		\node [style=textbox] (99) at (-3.75, 0.25) {$\alpha^{23}$};
		\node [style=none] (100) at (-4.75, 7.25) {};
		\node [style=none] (101) at (-4.75, 10.25) {};
		\node [style=none] (102) at (-4.75, 6) {};
		\node [style=none] (103) at (-4.75, 4.75) {};
		\node [style=none] (104) at (-4.75, 3.5) {};
		\node [style=none] (105) at (-4.75, 0.25) {};
		\node [style=none] (106) at (-0.25, 0.25) {};
		\node [style=none] (107) at (3.5, 1.25) {};
		\node [style=none] (108) at (3.5, -0.75) {};
		\node [style=yellownode] (109) at (1.75, 0.25) {$E^{11\ 23}_{24\ 25}$};
		\node [style=textbox] (110) at (6.75, 1.25) {$-^{25}$};
		\node [style=textbox] (111) at (4.5, 1.25) {$\alpha^{23}$};
		\node [style=textbox] (112) at (6.75, -0.75) {H$^{11}$};
		\node [style=textbox] (113) at (4.5, -0.75) {$+^{24}$};
		\node [style=none] (114) at (11.5, 2.25) {};
		\node [style=none] (115) at (7.75, 3.5) {};
		\node [style=none] (116) at (7.75, 1.25) {};
		\node [style=yellownode] (117) at (9.5, 2.25) {$\bar E^{10\ 23}_{12\ 25}$};
		\node [style=none] (118) at (1.5, 3.5) {};
		\node [style=textbox] (119) at (12.5, 2.25) {$\alpha^{22}$};
		\node [style=textbox] (120) at (14.75, 2.25) {O$^{10}$};
		\node [style=textbox] (121) at (16.75, 2.25) {$\alpha^{23}$};
		\node [style=none] (122) at (-9.5, 12.5) {};
		\node [style=none] (123) at (-0.5, 7.25) {};
		\node [style=none] (124) at (-0.5, 10.25) {};
		\node [style=none] (125) at (-0.5, 6) {};
		\node [style=none] (126) at (1.5, 4.75) {};
		\node [style=none] (127) at (8, -0.75) {};
		\node [style=none] (128) at (17.75, 2.25) {};
		\node [style=none] (129) at (17.75, 4.75) {};
		\node [style=none] (130) at (17.75, 6) {};
		\node [style=none] (131) at (17.75, 7.25) {};
		\node [style=none] (132) at (17.75, 10.25) {};
		\node [style=none] (133) at (17.75, 12.5) {};
		\node [style=rednode] (134) at (20.5, 7.5) {$\bar C$};
		\node [style=none] (135) at (22.5, 7.5) {};
		\node [style=none] (136) at (16.25, -0.75) {};
		\node [style=none] (137) at (20, 1.75) {};
		\node [style=none] (138) at (18.25, 0.5) {};
		\node [style=none] (139) at (22.5, 1.75) {};
	\end{pgfonlayer}
	\begin{pgfonlayer}{edgelayer}
		\draw [style=edge] (50) to (52);
		\draw [style=edge] (50) to (57);
		\draw [style=edge] (50) to (58);
		\draw [style=edge] (59) to (61);
		\draw [style=edge] (59) to (60);
		\draw [style=doubedge] (59) to (63);
		\draw [style=edge] (64) to (61);
		\draw [style=edge] (50) to (51);
		\draw [style=doubedge] (65) to (61);
		\draw [style=edge] (67) to (66);
		\draw [style=edge] (72) to (73);
		\draw [style=edge] (66) to (78);
		\draw [style=edge] (78) to (77);
		\draw [style=edge] (66) to (68);
		\draw [style=edge] (68) to (79);
		\draw [style=edge] (80) to (66);
		\draw [style=rededge] (82.center) to (81);
		\draw [style=edge] (84) to (83);
		\draw [style=edge] (83) to (86);
		\draw [style=edge] (83) to (85);
		\draw [style=edge] (87) to (83);
		\draw [style=edge] (89) to (88);
		\draw [style=edge] (91) to (90);
		\draw [style=edge] (92) to (93);
		\draw [style=edge] (93) to (94);
		\draw [style=edge] (95) to (97);
		\draw [style=edge] (97) to (96);
		\draw [style=edge] (99) to (98);
		\draw [style=rededge, bend left, looseness=0.50] (105.center) to (81);
		\draw [style=rededge, bend left=15, looseness=0.50] (104.center) to (81);
		\draw [style=rededge, bend left=15, looseness=0.25] (103.center) to (81);
		\draw [style=rededge] (102.center) to (81);
		\draw [style=rededge, bend right=15, looseness=0.50] (100.center) to (81);
		\draw [style=rededge, bend right=15, looseness=0.50] (101.center) to (81);
		\draw [style=rededge] (106.center) to (109);
		\draw [style=rededge, bend left=15] (109) to (107.center);
		\draw [style=rededge, bend right=15] (109) to (108.center);
		\draw [style=edge] (111) to (110);
		\draw [style=edge] (113) to (112);
		\draw [style=rededge] (114.center) to (117);
		\draw [style=rededge, bend right=15] (117) to (115.center);
		\draw [style=rededge, bend left=15] (117) to (116.center);
		\draw [style=rededge] (118.center) to (115.center);
		\draw [style=edge] (119) to (120);
		\draw [style=edge] (120) to (121);
		\draw [style=rededge] (122.center) to (133.center);
		\draw [style=rededge] (124.center) to (132.center);
		\draw [style=rededge] (123.center) to (131.center);
		\draw [style=rededge] (125.center) to (130.center);
		\draw [style=rededge] (126.center) to (129.center);
		\draw [style=rededge, bend left=15, looseness=0.50] (133.center) to (134);
		\draw [style=rededge, bend left=15, looseness=0.75] (132.center) to (134);
		\draw [style=rededge] (131.center) to (134);
		\draw [style=rededge, bend right=15, looseness=0.50] (130.center) to (134);
		\draw [style=rededge, bend right=15, looseness=0.75] (129.center) to (134);
		\draw [style=rededge, bend right=15, looseness=0.50] (128.center) to (134);
		\draw [style=rededge] (134) to (135.center);
		\draw [style=rededge] (127.center) to (136.center);
		\draw [style=rededge] (137.center) to (139.center);
		\draw [style=rededge, bend right, looseness=0.75] (137.center) to (138.center);
		\draw [style=rededge, bend left, looseness=0.75] (138.center) to (136.center);
	\end{pgfonlayer}
\end{tikzpicture}}.
\end{center}
\end{example}
\begin{example}
The morphism in Example~\ref{ex:mdisc} is sent to the following reaction by $R$: the chemical subgraphs are given by $\{w,z,u,v\}$ and $\{w,z,u,v,e,f\}$, while the labelled bijection and the pre-chemical graph isomorphism both send a vertex to itself.
\end{example}

\begin{definition}[Retrosynthetic step]\label{def:retro-step}
A {\em retrosynthetic step} consists of
\begin{itemize}
\item a molecular graphs $T$ and $B$, called the {\em target}, and the {\em byproduct},
\item a finite set of molecular entities $M$, called the {\em environment},
\item a chemical graph $S$, whose connected components are called the {\em synthons},
\item a molecular graph $E$, whose components are called the {\em synthetic equivalents},
\item morphisms $d\in\Disc(T,S)$, $m\in\MMatch(S,E)$, $r\in\MReact(E,T+B)$.
\end{itemize}
\end{definition}
\begin{proposition}\label{prop:step-graphically}
The data of a retrosynthetic step are equivalent to existence of the following morphism (1-cell) in the layered prop generated by the diagram~\eqref{eq:m-functors}:
\begin{center}
\scalebox{.85}{\begin{tikzpicture}
	\begin{pgfonlayer}{nodelayer}
		\node [style=none] (0) at (-9, -1.5) {};
		\node [style=none] (1) at (-9, 1.5) {};
		\node [style=none] (2) at (10.75, 1.5) {};
		\node [style=none] (3) at (10.75, -1.5) {};
		\node [style=none] (4) at (-5, 1.5) {};
		\node [style=none] (5) at (-5, -1.5) {};
		\node [style=none] (6) at (2.5, 1.5) {};
		\node [style=none] (7) at (2.5, -1.5) {};
		\node [style=none] (8) at (5, 1.5) {};
		\node [style=none] (9) at (5, -1.5) {};
		\node [style=none] (10) at (-7, 2) {$\eset\mhyphen\Disc$};
		\node [style=none] (11) at (-5, 2) {$\refine$};
		\node [style=none] (12) at (0, 2) {$\MMatch$};
		\node [style=none] (13) at (2.75, 2) {$\refine_D$};
		\node [style=none] (14) at (5, 2) {$\coarsen$};
		\node [style=none] (15) at (8.25, 2) {$\MReact$};
		\node [style=redbox] (16) at (-7, 0) {$d$};
		\node [style=redbox] (17) at (-0.25, 0) {$m$};
		\node [style=redbox] (18) at (8, 0) {$r$};
		\node [style=none] (19) at (-9, 0) {};
		\node [style=none] (20) at (10.75, 0.75) {};
		\node [style=none] (21) at (-8.25, 0.5) {$T$};
		\node [style=none] (22) at (-4, 0.5) {$S$};
		\node [style=none] (23) at (3.75, 0.5) {$E$};
		\node [style=none] (24) at (9.5, 1) {$T$};
		\node [style=none] (25) at (-3, 1.5) {};
		\node [style=none] (26) at (-3, -1.5) {};
		\node [style=none] (27) at (-2.75, 2) {$\coarsen_D$};
		\node [style=none] (28) at (-4, 3.25) {$\MDisc$};
		\node [style=none] (29) at (3.75, 3.25) {$\MDisc$};
		\node [style=none] (30) at (-4, 2.75) {};
		\node [style=none] (31) at (-4, 1) {};
		\node [style=none] (32) at (3.75, 1) {};
		\node [style=none] (33) at (3.75, 2.75) {};
		\node [style=none] (38) at (5.75, 1.5) {};
		\node [style=none] (39) at (5.75, -1.5) {};
		\node [style=none] (40) at (6, 2) {$\refine_R$};
		\node [style=none] (41) at (10.75, -0.75) {};
		\node [style=none] (42) at (9.5, -1) {$B$};
	\end{pgfonlayer}
	\begin{pgfonlayer}{edgelayer}
		\draw [style=rededge] (19.center) to (16);
		\draw [style=rededge] (16) to (17);
		\draw [style=rededge] (17) to (18);
		\draw [style=rededge, bend left=15, looseness=0.75] (18) to (20.center);
		\draw [style=edge] (1.center) to (2.center);
		\draw [style=edge] (2.center) to (3.center);
		\draw [style=edge] (3.center) to (0.center);
		\draw [style=edge] (0.center) to (1.center);
		\draw [style=edge] (4.center) to (5.center);
		\draw [style=edge] (6.center) to (7.center);
		\draw [style=edge] (8.center) to (9.center);
		\draw [style=edge] (25.center) to (26.center);
		\draw [style=gray dashed] (30.center) to (31.center);
		\draw [style=gray dashed] (33.center) to (32.center);
		\draw [style=edge] (38.center) to (39.center);
		\draw [style=rededge, bend right=15, looseness=0.75] (18) to (41.center);
	\end{pgfonlayer}
\end{tikzpicture}}.
\end{center}
\end{proposition}
The morphism in the above proposition should be compared to the informal diagram in Figure~\ref{fig:clayden}. The immediate advantage of presenting a retrosynthetic step as a morphism in a layered prop is that it illuminates how the different parts of the definition fit together in a highly procedural manner. Equally importantly, this presentation is fully compositional: one can imagine performing several steps in parallel, or dividing the tasks of finding the relevant morphisms (e.g.~between different computers). Moreover, one can reason about different components of the step while preserving a precise mathematical interpretation (so long as one sticks to the rewrites (2-cells) of the layered prop).

\begin{definition}[Retrosynthetic sequence]\label{def:retro-sequence}
A {\em retrosynthetic sequence} for a target molecular entity $T$ is a sequence of morphisms $r_1\in M_1\mhyphen\React(E_1,T+B_0)$, $r_2\in M_2\mhyphen\React(E_2,E_1+B_1)$, \dots, $r_n\in M_1\mhyphen\React(E_n,E_{n-1}+B_{n-1})$ such that the domain of $r_i$ is a connected subgraph of the codomain of $r_{i+1}$:
\begin{center}
\scalebox{.85}{\begin{tikzpicture}
	\begin{pgfonlayer}{nodelayer}
		\node [style=none] (78) at (-5, 1.5) {};
		\node [style=none] (79) at (-5, -1.5) {};
		\node [style=none] (80) at (-10, 1.5) {};
		\node [style=none] (81) at (-10, -1.5) {};
		\node [style=none] (83) at (-7.75, 2) {$M_n\mhyphen\React$};
		\node [style=redbox] (84) at (-7.75, 0) {$r_n$};
		\node [style=none] (88) at (-10, 0) {};
		\node [style=none] (89) at (-9.25, 0.5) {$E_n$};
		\node [style=none] (90) at (-6.75, 1) {$E_{n-1}$};
		\node [style=none] (91) at (-3.5, 0) {$\cdots$};
		\node [style=none] (92) at (9.5, 1.5) {};
		\node [style=none] (93) at (9.5, -1.5) {};
		\node [style=none] (94) at (4.5, 1.5) {};
		\node [style=none] (95) at (4.5, -1.5) {};
		\node [style=none] (96) at (6.75, 2) {$M_1\mhyphen\React$};
		\node [style=redbox] (97) at (6.75, 0) {$r_1$};
		\node [style=none] (99) at (4.5, 0) {};
		\node [style=none] (100) at (5.25, 0.5) {$E_1$};
		\node [style=none] (102) at (3, 1.5) {};
		\node [style=none] (103) at (3, -1.5) {};
		\node [style=none] (104) at (-2, 1.5) {};
		\node [style=none] (105) at (-2, -1.5) {};
		\node [style=none] (106) at (0.25, 2) {$M_2\mhyphen\React$};
		\node [style=redbox] (107) at (0.25, 0) {$r_2$};
		\node [style=none] (109) at (-2, 0) {};
		\node [style=none] (110) at (-1.25, 0.5) {$E_2$};
		\node [style=none] (111) at (1.5, 1) {$E_1$};
		\node [style=none] (113) at (9.5, 0.75) {};
		\node [style=none] (114) at (8.25, 1) {$T$};
		\node [style=none] (115) at (9.5, -0.75) {};
		\node [style=none] (116) at (8.25, -1) {$B_0$};
		\node [style=none] (118) at (3, 0.75) {};
		\node [style=none] (120) at (3, -0.75) {};
		\node [style=none] (121) at (1.75, -1) {$B_1$};
		\node [style=none] (124) at (-5, 0.75) {};
		\node [style=none] (125) at (-5, -0.75) {};
		\node [style=none] (126) at (-6.25, -1) {$B_{n-1}$};
	\end{pgfonlayer}
	\begin{pgfonlayer}{edgelayer}
		\draw [style=rededge] (88.center) to (84);
		\draw [style=edge] (78.center) to (79.center);
		\draw [style=edge] (80.center) to (81.center);
		\draw [style=edge] (80.center) to (78.center);
		\draw [style=edge] (81.center) to (79.center);
		\draw [style=rededge] (99.center) to (97);
		\draw [style=edge] (92.center) to (93.center);
		\draw [style=edge] (94.center) to (95.center);
		\draw [style=edge] (94.center) to (92.center);
		\draw [style=edge] (95.center) to (93.center);
		\draw [style=rededge] (109.center) to (107);
		\draw [style=edge] (102.center) to (103.center);
		\draw [style=edge] (104.center) to (105.center);
		\draw [style=edge] (104.center) to (102.center);
		\draw [style=edge] (105.center) to (103.center);
		\draw [style=rededge, bend left=15, looseness=0.75] (97) to (113.center);
		\draw [style=rededge, bend right=15, looseness=0.75] (97) to (115.center);
		\draw [style=rededge, bend left=15, looseness=0.75] (107) to (118.center);
		\draw [style=rededge, bend right=15, looseness=0.75] (107) to (120.center);
		\draw [style=rededge, bend left=15, looseness=0.75] (84) to (124.center);
		\draw [style=rededge, bend right=15, looseness=0.75] (84) to (125.center);
	\end{pgfonlayer}
\end{tikzpicture}}.
\end{center}
\end{definition}
Thus a retrosynthetic sequence is a chain of reactions, together with reaction environments, such that the products of one reaction can be used as the reactants for the next one, so that the reactions can occur one after another (assuming that the products can be extracted from the reaction environment, or one environment transformed into another one). In the formulation of a generic retrosynthesis procedure below, we shall additionally require that each reaction in the sequence comes from ``erasing'' everything but the rightmost cell in a retrosynthetic step.

We are now ready to formulate step-by-step retrosynthetic analysis. The procedure is a high-level mathematical description that, we suggest, is flexible enough to capture all instances of retrosynthetic algorithms. As a consequence, it can have various computational implementations. Let $T$ be some fixed molecular entity. We initialise by setting $i=0$ and $E_0\coloneqq T$.
\begin{enumerate}
\item Choose a subset $\mathcal D$ of disconnection rules,
\item Provide at least one of the following:
\begin{enumerate}
\item a finite set of reaction schemes $\mathcal S$,
\item a function $\mathfrak F$ from molecular graphs to finite sets of molecular graphs,
\end{enumerate}
\item Search for a retrosynthetic step with $d\in\eset\mhyphen\Disc(E_i,S)$, $m\in\MMatch(S,E)$, and $r\in\MReact(E,E_i+B_i)$ such that all disconnection rules in $d$ and $D(m)$ are in $\mathcal D$, and we have at least one of the following:
\begin{enumerate}
\item there is an $s\in\mathcal S$ such that the reaction $r$ is an instance of $s$,
\item $E_i + B_i\in\mathfrak F(E)$;
\end{enumerate}
if successful, set $E_{i+1}\coloneqq E$, $M_{i+1}\coloneqq M$, $r_{i+1}\coloneqq r$ and proceed to Step 4; if unsuccessful, stop,
\item Check if the molecular entities in $E_{i+1}$ are known (commercially available): if yes, terminate; if no, increment $i\mapsto i+1$ and return to Step 1.
\end{enumerate}
Note how our framework is able to incorporate both template-based and template-free retrosynthesis, corresponding to the choices between (a) and (b) in Step 2: the set $\mathcal S$ is the template, while the function $\mathfrak F$ can be a previously trained algorithm, or other unstructured empirical model of reactions. We can also consider hybrid models by providing both $\mathcal S$ and $\mathfrak F$.

We take the output retrosynthetic sequence to always come with a specified reaction environment for each reaction. Currently existing tools rarely provide this information (mostly for complexity reasons), and hence, in our framework, correspond to the set $M$ always being empty in Step 3.

Steps 1 and 2 both require making some choices. Two approaches to reduce the number of choices, as well as the search space in Step 3, have been proposed in the automated retrosynthesis literature: to use molecular similarity~\cite{Coley2017-similarity}, or machine learning~\cite{Lin2020}. Chemical similarity can be used to determine which disconnection rules, reactions and environment molecules are actually tried: e.g.~in Step 1, disconnection rules that appear in syntheses of molecules similar to $T$ can be prioritised.

Ideally, each unsuccessful attempt to construct a retrosynthetic step in Step 3 should return some information on why the step failed: e.g.~if the codomain of a reaction fails to contain $E_i$, then the output should be the codomain and a measure of how far it is from $E_i$. Similarly, if several reactions are found in Step 3, some of which result in products $O$ that do not contain $E_i$, the step should suggest minimal alterations to $E$ such that these reactions do not occur. This can be seen as a {\em deprotection} step: the idea is that in the next iteration the algorithm will attempt to construct (by now a fairly complicated) $E$, but now there is a guarantee this is worth the computational effort, as this prevents the unwanted reactions from occurring ({\em protection} step). Passing such information between the layers would take the full advantage of the layered prop formalism.

\section{Discussion and Future Work}\label{sec:discussion}

The main conceptual contributions of formulating retrosynthesis in layered props are the explicit mathematical descriptions of retrosynthetic steps (Definition~\ref{def:retro-step}) and sequences (Definition~\ref{def:retro-sequence}), which allows for a precise formulation of the entire process, as well as of more fine-grained concepts. While in the current article we showed how to account for the available disconnection rules, reactions and environmental molecules, the general formalism of layered props immediately suggests how to account for other environmental factors (e.g.~temperature and pressure). Namely, these should be represented as posets which control the morphisms that are available between the chemical compounds. One idea for accounting for the available energy is via the disconnection rules: the higher the number of bonds that we are able to break in one step, the more energy is required to be present in the environment.

Apart from modelling retrosynthesis, another potential use of the reaction contexts is to capture context-dependent chemical similarity. While molecular similarity is a major research topic in computational chemistry~\cite{mol-similarity}, the current approaches are based on comparing the molecular structure (connectivity, number of rings etc.) of two compounds, and is therefore bound to ignore the reaction environment. Other advantages of our framework are representation of the protection-deprotection steps, and hard-wiring of chirality into the formalism.

At the level of the formalism, the next step is to model translations between the reaction environments as functors of the form $M\mhyphen\React\rightarrow N\mhyphen\React$. This would allow presenting a retrosynthetic sequence as a single, connected diagram, closely corresponding to actions to be taken in a lab. Similarly, we note that the informal algorithmic description in Section~\ref{sec:retrosynthesis} could be presented internally in a layered prop: Steps 1 and 2 amount to choosing subcategories of $\Disc$ and $\React$.

A theoretical issue that should be addressed in future work is the precise relation between reactions and disconnection rules. As was mentioned when introducing the disconnection rules, we believe that any reaction can be decomposed into a sequence of disconnection rules. This amounts to proving that the translation functor $R$ is full, hence giving a completeness result for reactions with respect to the disconnection rules. In this way, the reactions can be seen as providing semantics for the disconnection rules. This also has a practical significance from the point of view of algorithm design: it would show that all computations, in principle, could be done with just using the disconnection rules.

On the practical side, the crucial next step is to take existing retrosynthesis algorithms and encode them in our framework. This requires implementing the morphisms of the layered prop in the previous section in some software. As the morphisms are represented by string diagrams, one approach is to use proof formalisation software specific to string diagrams and their equational reasoning, such as~\cite{cartographer}. Alternatively, these morphisms could be coded into a programming language like python or Julia. The latter is especially promising, as there is a community writing category-theoretic modules for it~\cite{AlgebraicJulia}. As a lower level description, the disconnection rules and the reactions presented could be encoded in some graph rewriting language, such as Kappa~\cite{kappa-language,kappa2004,krivine-siglog,rewriting-life21}, which is used to model systems of interacting agents, or M\O{}D~\cite{mod-language,intermediate-level,chem-trans-motifs,rewriting-life21}, which represents molecules as labelled graphs and generating rules for chemical transformations as spans of graphs (akin to this work). In order to formally represent reactions as disconnection rules, we need to rewrite string diagrams, the theory for which has been developed in a recent series of articles~\cite{sdrt1,sdrt2,sdrt3}.

\paragraph{Acknowledgements.} F.~Zanasi acknowledges support from \textsc{epsrc} EP/V002376/1.

\bibliographystyle{plain}
\bibliography{bibliography}

\end{document}